\documentclass[10pt]{amsart}
\usepackage[utf8]{inputenc}
\usepackage{amscd}
\usepackage{amssymb}
\usepackage{pstricks}
\usepackage{graphicx}
\usepackage{float}
\usepackage{todonotes}
\usepackage{subcaption}
\usepackage{xcolor}
\usepackage{textcomp}
\usepackage{cancel}
\usepackage{forest}
\usepackage[mathlines]{lineno}
\usepackage{url}
\usepackage{tikz}
\usetikzlibrary{shapes.geometric, arrows, positioning}

\tikzstyle{startstop} = [rectangle, rounded corners, minimum width=2.5cm, minimum height=0.8cm, text centered, draw=black, fill=blue!55, font=\small]
\tikzstyle{process} = [rectangle, minimum width=2.5cm, minimum height=0.5cm, text centered, draw=black, fill=blue!30, font=\small]
\tikzstyle{decision} =  [rectangle, minimum width=2.5cm, minimum height=0.5cm, text centered, draw=black, fill=blue!10, font=\small]
\tikzstyle{arrow} = [thick,->,>=stealth]

\newtheorem{theorem}{Theorem}
\newtheorem{proposition}{Proposition}
\newtheorem{lemma}{Lemma}
\newtheorem{definition}{Definition}
\newtheorem{corollary}{Corollary}

\newtheorem{remark}{Remark}

\numberwithin{equation}{section}
\numberwithin{corollary}{section}
\numberwithin{proposition}{section}
\numberwithin{theorem}{section}
\numberwithin{lemma}{section}

\title[A metapopulation almost periodic model]{Persistence/extinction scenarios in an almost periodic metapopulation with competition and habitat destruction}

\author[Amster]{Pablo Amster}
\author[Robledo]{Gonzalo Robledo}
\author[Sepulveda]{Daniel Sep\'ulveda}
\email{pamster@dm.uba.ar,grobledo@uchile.cl,daniel.sepulveda@utem.cl}

\address{Departamento de Matem\'atica, Universidad de Buenos Aires, Argentina. IMAS -- CONICET.}
\address{Departamento de Matem\'aticas, Universidad de Chile. Las Palmeras 3425, \~Nu\~noa, Santiago, Chile.}
\address{Departamento de Matem\'atica, Universidad Tecnol\'ogica Metropolitana, \~Nu\~noa, Santiago, Chile.}

\begin{document}

\begin{abstract}
We study an almost periodic version of a metapopulation model developed by Tilman \textit{et.al} and Nee \textit{et.al} in the nineties, which generalizes the classical Levins approach by considering several species in competition a\-ffec\-ted by habitat destruction. The novelty is to assume that the colonization and extinction rates are positive almost periodic functions whereas our main results show that the predominance of either colonization or extinction forces of a specific species is equivalent to the property of exponential dichotomy of a scalar linear differential equation. By using well known results of exponential dichotomy theory, we carry out a recursive and exhaustive  description of persistence/extinction scenarios. In addition, we start 
a preliminary discussion describing a more elusive behavior when the colonization and extinction forces are similar in average.  
\end{abstract}
\maketitle 
\section{Introduction}
The loss of habitat has multiple causes: expansion of agricultural land, exploitation of natural resources (wood, minerals, fisheries, aquifers), overgrazing, bad ma\-na\-ge\-ment of toxic wastes, uncontrolled wildfires, unregulated real estate expansion, desertification, coral reef degradation, etc,
which have a dramatic impact on the biodiversity declining. In fact, the loss of habitat is considered as the major threat to biodiversity
\cite{Britton,Hanski-2011}. Nevertheless, a detailed description of these effects remains to be done as well as the
modeling of these issues has been addressed from several perspectives. Within this context, the metapopulation theory has provided a fruitful framework for the study of the persistence or extinction of diverse populations affected by progressive habitat degradation.
In this article, we study a metapopulation model introduced by  S. Nee \textit{et.al} \cite{Nee} 
and D. Tilman \textit{et.al} \cite{Tilman} which retains the main features of the classical model introduced by R. Levins.

The theory of metapopulations \cite{Gilpin,Hanski} assumes the existence of fragmented habitats where
any habitable portion is called a \textit{patch}. Now, all the species inhabiting a specific patch conform a \textit{local population} and --roughly speaking-- a \textit{metapopulation} is a collection of local populations in a region.

According to P. Chesson \cite{Chesson}, a classical metapopulation has the fo\-llowing properties:

\begin{itemize}
\item[$\bullet$] The patches are  partially isolated between them and are able of autosustaining for several generations in absence of colonization for other local populations. 
\item[$\bullet$] The local population extinction occurs in a time scale considering many generations.
\item[$\bullet$] The migration between local populations leads to reestablishment of local populations following local extinction.
\end{itemize}

The metapopulation modeling has two time scales: a short scale only assumes that any patch is either occupied or unoccupied, and the internal demographic dynamics (births and deaths) of the patches is neglected. On the other hand, the dynamical balance between the processes of colonization and extinction takes place in a time scale encompassing several generations of demographical growth.

\subsection{The Levins model with competition and habitat destruction}
The first generation of metapopulation models started with the seminal work of the mathematical ecologist Richard Levins \cite{Levins,Levins2},
who coined the concept of metapopulation and described it by the ordinary differential equation:
\begin{equation}
\label{levins-0}
p'=cp(1-p)-mp,
\end{equation}
where $p\in [0,1]$ is the fraction of occupied patches by a single species. In addition $c>0$ is the colonization rate and $m>0$ is the extinction rate. If $p(0)\in (0,1]$ the solutions of (\ref{levins-0}) have the following
behavior, which is determined by the dominance of either colonization or extinction forces:
\begin{displaymath}
\lim\limits_{t\to +\infty}p(t)=\left\{
\begin{array}{ccl}
\displaystyle 0  & \textnormal{if}&   c \leq m \\
\displaystyle p^{*}:=1-(m/c)  &\textnormal{if}  & m<c. 
\end{array}\right.
\end{displaymath}

We point out that if $c\neq m$, the above equilibria $0$ and $p^{*}$ are hyperbolic and the convergence is exponential. On the other hand, a bifurcation occurs when $c=m$, namely, these equilibria collapse at zero and there is neither hyperbolicity nor exponential convergence.

The generalization of the Levins model has been carried out by a myriad of approaches and we refer the reader to \cite{Etienne},\cite[Ch.4]{Hanski},\cite[Table 1]{Marquet} 
for a detailed discussion and bibliography. In this article we will be focused on a specific extension
which considers a metapopulation encompassing $n$ species in competition. For example, the following 
model is considered in \cite{Case}: 
\begin{equation}
\label{gene2}
\left\{\begin{array}{rcl}
p_{1}'&=&c_{1}p_{1}(1-p_{1})-m_{1}p_{1}\\
p_{2}'&=&c_{2}p_{2}(1-p_{1}-p_{2})-m_{2}p_{2}-c_{1}p_{1}p_{2},
\end{array}\right.
\end{equation}
where $p_{i}\in [0,1]$ describe the fraction of the occupied patches by the
$i$--th species. As before, $c_{i}$ and $m_{i}$, respectively, are the colonization and extinction rates of each species. The system (\ref{gene2}) deserves some remarks:
\begin{itemize}
\item[a)] The abundance of the species 1 is  described by an equation similar to (\ref{levins-0})
since it is not affected by the other one. 
\item[b)] The species 2 is unable to migrate a patch occupied by the species 1. Moreover,
if the species 1 invades a patch occupied by the species 2, this last species is immediately driven to extinction. 
\item[c)] By gathering the above statements, the species 1 and 2 are respectively called \textit{superior}
and \textit{inferior} competitors.
\item[d)] It is assumed that the colonization abilities of the species 2 are better than those of the species 1, that is $c_{1}< c_{2}$. For this reason, the species 2 is also called a \textit{fugitive species}.
\end{itemize}

A generalization of (\ref{gene2}) carried out in \cite{Nee} considered habitat destruction, that is,
the existence of a fraction $h\in (0,1)$ of habitable patches, which leads to
\begin{equation}
\label{gene2Tilman}
\left\{\begin{array}{rcl}
p_{1}'&=&c_{1}p_{1}(h-p_{1})-m_{1}p_{1}\\
p_{2}'&=&c_{2}p_{2}(h-p_{1}-p_{2})-m_{2}p_{2}-c_{1}p_{1}p_{2},
\end{array}\right.
\end{equation}
and a generalization for $n$ species was carried out in \cite{Tilman}:
\begin{equation}
\label{eq:1}
p_i'=c_i p_i\left(h-\sum_{j=1}^{i}p_j \right)-m_ip_i-\sum_{j=1}^{i-1}c_j p_i p_j,
\end{equation}
where $p_i$ represents the proportion of patches occupied by the $i$--th species with $i\in \{1,\ldots,n\}$, $c_i$ and $m_{i}$ are respectively the species-specific colonization and extinction rates. Similarly as in \cite{Case} and \cite{Nee}, a fundamental assumption in \cite{Tilman} 
is the \textit{competition--colonization trade--off}:
\begin{equation}
\label{trade-off}
c_{1}< c_{2} < \ldots < c_{n},
\end{equation}
that is, the colonization rate of the species tends to increase as its competitiveness declines. This leads to the fact that
the best competitor is the worst colonizer and vice versa. It is important to notice that the hierarchical structure of (\ref{eq:1}) allows 
an easy identification of all its equilibria. In fact, by following an idea of \cite{Miller}, we can rewrite (\ref{eq:1})
as the generalized Lotka--Volterra system:
\begin{equation}
\label{GLV}
p_{i}'=p_{i}\left(r_{i}(h)-\sum\limits_{j=1}^{n}A_{ij}p_{j}\right)    \quad \textnormal{with $i=1,\ldots,n$},
\end{equation}
where 
\begin{displaymath}
r_{i}(h)=hc_{i}-m_{i} \quad \textnormal{and} \quad
A_{ij}=\left\{\begin{array}{lcr}
c_{i}+c_{j} &\textnormal{if} & j<i \\
c_{i} & \textnormal{if}&  i=j \\
0 & \textnormal{if}& j>i.
\end{array}\right.
\end{displaymath}

By defining $A=\{A_{ij}\}_{i,j=1}^{n}$ and $r(h)=(r_{1}(h),\ldots,r_{n}(h))^{T}$, 
a global study of (\ref{GLV}) can be addressed for some configurations of the parameters space that give rise to $2^{n}$ possible hyperbolic equilibria which describe scenarios ranging from to persistence of all species --the equilibrium $p_{1}^{\ast}=A^{-1}r(h)$-- to the total extinction of the species, namely, the equilibrium $p_{2^{n}}^{\ast}=(0,\ldots,0)^{T}$.  On other hand, the intermediate persistence/extinction scenarios are described by the positive/null components of the asymptotically stable equilibria $p_{i}^{\ast}=(p_{i1},\ldots,p_{in})^{T}$ with $i\in \{2,3,\ldots,2^{n}-1\}$.

Furthermore, for restricted configurations of the space of parameters, a group of the above equilibria collapses into a non--hyperbolic equilibrium; this happens when the colonization forces of a specific competing species are equal to the extinction ones.

\subsection{The Levins model with almost periodic colonization and extinction rates} 
Among the so many generalizations of the Levins model, the assumption that -from a time scale encompassing several generations- the colonization and extinction rates can be considered constant has been weakened by using 
a stochastic approach. In fact, Levins assumes in \cite{Levins} that the extinction rate depends
randomly  on the frequency of a gene in the population
whose average is $\overline{m}$ and its variance is $\sigma_{m}^{2}$. When $c>\overline{m}+\sigma_{m}^{2}$ it is proved the existence of
a modal distribution with a peak at $p=1-(\overline{m}+\sigma_{m}^{2})/c$. Nevertheless, 
to the best of our knowledge, there are no generalizations from a deterministic point of view.

In this context, we will consider the Tilman's model (\ref{eq:1}) with time varying colonization and extinction rates:  
\begin{equation}
\label{tilman1-TV}
p_{i}'=c_i(t) p_i\left(h-\sum_{j=1}^{i}p_j \right)-m_i(t)p_i-\sum_{j=1}^{i-1}c_j(t) p_i p_j, \quad
\textnormal{for $i=1,\ldots,n$}.
\end{equation}

A crucial assumption of this work is that the colonization and extinction rates will be supposed to be positive almost periodic functions $c_{i},m_{i}\colon \mathbb{R}\to (0,+\infty)$ satisfying the additional condition:
\begin{equation}
\label{infap}
\inf\limits_{t\in (-\infty,+\infty)}c_{i}(t)=c_{i}^{-}>0.
\end{equation}

Let us recall that a continuous function
$v\colon \mathbb{R}\to \mathbb{R}$ is almost periodic in the sense of H. Bohr if, for any $\varepsilon>0$, there exists a number 
$\ell_{\varepsilon}>0$ such that any interval
of length $\ell_{\varepsilon}$ contains a number $\tau$ such that $|v(t+\tau)-v(t)|<\varepsilon$ for any $t\in \mathbb{R}$ and we refer the reader to \cite{Besicovitch,Bohr,Fink} for a deeper treatment. 
{For our purposes it will be enough to recall that} 
almost periodic functions are bounded, uniformly continuous and have noticeable statistic properties. For example, its \textit{average}
is well defined by the following limit:
\begin{equation}
\label{average-a}
\mathcal{M}[v]:=\lim\limits_{L\to +\infty}\frac{1}{L}\int_{c}^{L+c}v(s)\,ds \quad \textnormal{uniformly for any $c\in \mathbb{R}$},
\end{equation}
and we refer to \cite[pp. 39--44]{Bohr} for details.

Furthermore, it will be assumed that the \textit{competition--colonization trade--off} is satisfied from an average point of view:
\begin{equation}
\label{BCWC}
0<\mathcal{M}[c_{1}] <\mathcal{M}[c_{2}] < \cdots < \mathcal{M}[c_{n}],    
\end{equation}
that is, the average of the colonization rate of the species tends to increase as its competitiveness declines.

It is worth to recall that, as the continuous $T$--periodic functions are a particular case of almost periodic functions, 
it follows that if $c_{i}(\cdot)$ is continuous and $T_{i}$--periodic, the above average takes a simpler form:
\begin{equation}
\label{average}
\mathcal{M}[c_{i}]:=\lim\limits_{L\to +\infty}\frac{1}{L}\int_{0}^{L}c_{i}(s)\,ds = \frac{1}{T_{i}}\int_{0}^{T_{i}}c_{i}(s)\,ds.   
\end{equation}

In addition, in case that the colonization rates $c_{i}(\cdot)$ are $T_{i}$--periodic, the almost periodic framework will be
useful when $T_{k}/T_{j}\notin \mathbb{Q}$ for some couple $k,j\in \{1,\ldots,n\}$.

\medskip 

The almost periodic system (\ref{tilman1-TV}) deserves additional comments from a modeling, ecological and mathematical perspectives:

$\bullet$ From an ecological point of view, the system is \textit{competitive}, that is, the growth of a species has a negative impact on the growth of any other species. From a mathematical point of view, this corresponds to the fact that all the off diagonal entries of the corresponding Jacobian matrix are negative or zero.

$\bullet$ From an ecological point of view, the system is \textit{hierarchical}, that is, the species are ranked from $i=1$
(best) to $i=n$ (worst) according to its competitive ability.  Notice that the dynamics of the $j$--th species is not affected by the $k$--th species with $k>j$. 
From a mathematical point of view, the system is triangular, that is, its corresponding Jacobian matrix is lower triangular. 

$\bullet$ From a modeling point of view, the autonomous system (\ref{eq:1}) has a strong assumption: the colonization and extinction rates are constant in big time scales. A more realistic assumption could be
to assume either deterministic or randomly time varying rates. In this modeling context, we studied the thought provoking works of M. Allais \cite{Allais} and J. Bass \cite{Bass}, in which it is explained how the almost periodic and another functions whose averages are well defined can be used in the simulation of random behavior. Furthermore, it is also important to emphasize that the almost periodic functions have not been considered in depth
for theoretical ecologists, despite of that, we think that they also provide interesting tools for describing cyclical behaviors 
beyond the periodic one.

$\bullet$ From an ecological point of view, as in the autonomous case, the inferior competitors cannot colonize the patches occupied by the supperior competitors. This arises the following invariance result:
\begin{lemma}
\label{L1}
Let $t\mapsto p(t)=(p_{1}(t),\ldots,p_{n}(n))$ be a solution of the system \eqref{tilman1-TV}. If 
$\sum\limits_{i=1}^{n}p_{i}(t_{0})\in (0,h)$ then $\sum\limits_{i=1}^{n}p_{i}(t)\in (0,h)$ for any $t\geq t_{0}$.
\end{lemma}

\begin{proof}
See Appendix A.
\end{proof}

An essential difference with the autonomous case is that a qualitative study of a nonautonomous system as (\ref{tilman1-TV}) cannot be addressed from an equilibria based approach since the persistence of some species could be not achieved under the form of an attractive equilibrium. This motivates to introduce the following definition:
\begin{definition}
\label{GBS}
A \textbf{non autonomous equilibrium} of the system \eqref{tilman1-TV} is a solution  $\phi^{\ast}=(\phi_{1}^{\ast},\ldots,\phi_{n}^{\ast})\colon \mathbb{R}\to \mathbb{R}^{n}$ such that
\begin{displaymath}
\sup\limits_{t\in (-\infty,+\infty)}|\phi_{1}^{\ast}(t)|+\ldots+|\phi_{n}^{\ast}(t)|<+\infty.    
\end{displaymath}
\end{definition}

We stress that recent advances in the theory of non autonomous dynamical systems \cite{Kloeden2} have pointed out 
that, instead to focus in the classical \textit{static} equilibria, a global theory of non autonomous systems
can be carried out on the basis of the non autonomous equilibria, also called \textit{globally bounded solutions}. In particular, 
our main results will classify persistence/extinction scenarios described in terms of globally attractive non autonomous equilibria.

Until now, our description of metapopulation models has used intuitive notions of persistence and extinction. In order
to study the asymptotic behavior of (\ref{tilman1-TV}) we will work with the following formal definitions:
\begin{definition}
\label{DefUP}
The $i$--th species described in \eqref{tilman1-TV} is:

\noindent \textnormal{a)} Strongly uniformly persistent if there exists $\delta_{i}>0$ such that the fraction of patches inhabited by the $i$--th species verifies
\begin{equation*}
\liminf\limits_{t\to +\infty}p_{i}(t)\geq \delta_{i} \mbox{ for any non null initial condition.}
\end{equation*}

\noindent \textnormal{b)} Weakly persistent if the fraction of patches inhabited by the $i$--th species verifies
\begin{equation*}
\limsup\limits_{t\to +\infty}p_{i}(t)>0 \mbox{ for any non null initial condition.}
\end{equation*}

\noindent \textnormal{c)} Strongly driven to the extinction if the fraction of patches inhabited by the $i$--th species verifies
\begin{equation*}
\limsup\limits_{t\to +\infty}p_{i}(t)=0 \mbox{ for any non null initial condition.}
\end{equation*}

\noindent \textnormal{d)} Weakly driven to the extinction if the fraction of patches inhabited by the $i$--th species verifies
\begin{equation*}
\liminf\limits_{t\to +\infty}p_{i}(t)=0 \mbox{ for all non null initial condition.}
\end{equation*}
\end{definition}

We point out that statements a) and b) from Definition \ref{DefUP} are inspired in \cite[pp.175--176]{Thieme} 
and \cite[Def.3.4.2]{Zhao}. On the other hand, the statements c) and d) are not standard in the literature but they allow highlighting that a species can be weakly persistent and
at the same time weakly driven to the extinction.

We will prove that the strong uniform persistence/extinction scenarios of (\ref{tilman1-TV}) will be achieved under the form of a globally attractive positive almost periodic equilibria. In order to describe this convergence, we will use the little--o Landau's notation $f(t)=\textit{o}(1)$ to 
say that $f(t)\to 0$ when $t\to +\infty$. In addition, we will say that
$f(t)$ converges exponentially to zero if $f(t)=\textit{o}(1)$ and there exist $\gamma>0$ and a continuous map
$\mathcal{K}\colon  [0,+\infty)\to (0,+\infty)$ such that:
\begin{equation}
\label{NUEC}
|f(t)|\leq \mathcal{K}(t_{0})e^{-\gamma (t-t_{0})}|f(t_{0})| \quad \textnormal{for any $t\geq t_{0}$}.    
\end{equation}

We stress that the convergence described by (\ref{NUEC}) depends on the elapsed time between 
$t$ and $t_{0}$ but also on 
the initial time $t_{0}$. This type of convergence
is called  \textit{nonuniform exponential convergence} in the literature since the decay is exponential but it is not uniform with respect to the initial time $t_{0}$. 

\subsection{Structure of the article}
This article is on the intersection of theoretical ecology and the qualitative theory of nonautonomous differential equations. Moreover, the proof of our main results use methods and ideas which are current tools in nonautonomous dynamics but are far to be well known in theoretical ecology. In order to alleviate  this problem, we 
{believe} useful to write 
{a thorough} 
Section 2 focused in providing a common basis for both research communities: the property of exponential dichotomy for the scalar case.

The sections 3 and 4 study the model (\ref{tilman1-TV}) with $n=1$ and $n=2$ respectively and describes the possible scenarios of strong persistence/ strong extinction for the metapopulation. The case $n=2$ illustrates a recursive procedure to determine necessary and sufficient conditions for strong persistence/extinction, which is briefly described in the section 5. It is important to emphasize that the above mentioned conditions
will be 
{characterized} 
in terms of the exponential dichotomy of a scalar linear equation combined
with $h\in (0,1)$ as a bifurcation parameter.

The section 6 is devoted to explore some limit cases for $n=1$ and $n=2$ where the strong persistence/extinction scenarios are no longer verified. The resulting behaviors show that the weak extinction is achieved while the weak persistence could be achieved for some cases. 

\section{Exponential dichotomy for scalar equations}

From now on it will be very useful to emphasize that the first equation of (\ref{tilman1-TV}) is of Bernoulli type, this implies that the change of variables $y=1/p_{1}$ transforms it in a non-homogeneous equation of type:
\begin{equation}
\label{sca0}
y'=\mathfrak{a}(t)y+\mathfrak{v}(t),
\end{equation}
whose linear part is the equation 
\begin{equation}
\label{sca1}
x'=\mathfrak{a}(t)x.
\end{equation}

The above mentioned change of variables combined with the hierarchical structure of the system (\ref{tilman1-TV}) will allow us to carry out a recursive study of the qualitative properties of all and each one of the equations of system (\ref{tilman1-TV}). Notice that
the behavior of the $i$--th species is determined by that of the competitively superior species, namely, the $j$--th species (with $j=1,\ldots,i-1$)
combined with the qua\-li\-ta\-tive properties of the $i$--th equation of (\ref{tilman1-TV}), which can be transformed into an equation of type (\ref{sca0}).

A crucial observation will be that, at every stage, the balance between colonization and extinction forces is described by linear equations of type (\ref{sca1}) and the predominance of one of these forces will be equivalent to the property of \textit{exponential dichotomy} of (\ref{sca1}). 

There exists few applications of the exponential dichotomy theory in life sciences and we refer to \cite{GLK,Kloeden2} for seminal examples. Additional applications have been carried out for the study of Lotka--Volterra systems in \cite{GLK} and for sensitivity a\-na\-lysis and total effects in trophic chains in \cite{RRS}. This article can bee seen on this trend.

\subsection{Definition and properties}
The exponential dichotomy \cite{APM,Coppel,Kloeden,Lin2} is a property of linear time varying ODE systems
which emulates the hyperbolicity condition of the autonomous case. 
Surprisingly, there are few descriptions for the scalar case, which can be treated in a simple and self-contained manner. This is the goal of this section.

\begin{definition}
\label{DiEx}
The equation \eqref{sca1} has an exponential dichotomy on the 
unbounded interval $J$ if and only if there exist constants $C\ge 0$ and $\alpha>0$ such that
\begin{equation}
\label{dico}
\left|\int_s^t \mathfrak{a}(r)\,dr 
\right|\ge \alpha(t-s) - C\qquad \textnormal{for any $t\geq s$ with $t,s\in J$}.
\end{equation}
\end{definition}

By redefining the constant $C$ if necessary, it is clear that Definition \ref{DiEx} always corresponds to one of the following situations:

\begin{enumerate}
\item[a)] There  exist constants $C\ge 0$ and $\alpha>0$ such that
\begin{equation}
    \label{proj-id} 
\int_s^t \mathfrak{a}(r)\,dr  \le C - \alpha(t-s) \qquad \textnormal{for any $t\geq s$ with $t,s\in J$}.\end{equation}

\item[b)] 
There  exist constants $C\ge 0$ and $\alpha>0$ such that
\begin{equation}
    \label{proj-null}\int_t^s \mathfrak{a}(r)\,dr  \ge  \alpha(s-t) - C \qquad \textnormal{for any $s\geq t$ with $t,s\in J$}.
    \end{equation}
\end{enumerate}

The above properties are equivalent to the  standard definition of exponential dichotomy:
\begin{itemize}
\item[a')] There exist constants $K\ge 1$ and $\alpha>0$ such that 
\begin{equation}
\label{DEPI}
 e^{\int_{s}^{t}\mathfrak{a}(r)\,dr}\leq Ke^{-\alpha (t-s)} \quad \textnormal{for any $t\geq s$ with $t,s\in J$}.
 \end{equation}
\item[b')] There exist constants $K\ge 1$ and $\alpha>0$ such that 
\begin{equation}
\label{DENP}
e^{-\int_{t}^{s}\mathfrak{a}(r)\,dr}\leq Ke^{-\alpha (s-t)} \quad \textnormal{for any $s\geq t$ with $t,s\in J$}.
\end{equation} 
\end{itemize}

In fact, inspired in the standard definition, the case a) corresponds to the exponential dichotomy  with the \textit{identity projector} while the case b) corresponds to 
the exponential dichotomy  with the \textit{null projector}.

In general, the intervals usually considered in the literature are
$J=(-\infty,0]$, $J=[0,+\infty)$ and $J=(-\infty,+\infty)$. 

\begin{itemize}
\item[$\bullet$] If (\ref{sca1}) has an exponential dichotomy on $J=[0,+\infty)$, the properties (\ref{DEPI})--(\ref{DENP}) 
says that any solution either converges exponentially to zero (identity projector) or diverges exponentially (null projector) as $t\to +\infty$, since (\ref{DENP}) is equivalent to $K^{-1}e^{\alpha(t-t_0)}\leq e^{\int_{t_0}^t\mathfrak{a}(r)\,dr}$ with $t\geq t_{0}\geq 0$. 
\item[$\bullet$] Similarly, if (\ref{sca1}) has an exponential dichotomy on $J=(-\infty,0]$, the properties (\ref{DEPI})--(\ref{DENP}) 
says that any solution either converges exponentially to zero (null projector) or diverges exponentially (identity projector) as $t\to -\infty$, since (\ref{DEPI}) is equivalent to $K^{-1}e^{\alpha(t_0-t)}\leq e^{\int_{t_0}^{t}\mathfrak{a}(r)dr}$ with $t\leq t_0\leq 0$. 
\item[$\bullet$] 
If (\ref{sca1}) has an exponential dichotomy on $J=(-\infty,+\infty)$ then it also has an exponential dichotomy on $(-\infty,0]$
and $[0,+\infty)$ and inherits the above asymptotic behaviors at $\pm \infty$.
\end{itemize}

As we can see, if (\ref{sca1}) has an exponential dichotomy on $J$, its solutions have 
a \textit{dichotomic} asymptotic behavior at $t\to +\infty$ or $t\to -\infty$: they are either convergent or divergent at exponential rate. This qualitative asymptotic behavior is behind the name of exponential dichotomy.

\begin{remark}
\label{BST}
As a consequence of the previous discussion we can see that, if the equation \eqref{sca1} has an exponential dichotomy on $(-\infty,+\infty)$ then
the unique solution bounded on $(-\infty,+\infty)$ is the trivial one. In fact, any other solution will be divergent
at $+\infty$ or $-\infty$. 
\end{remark}

\begin{lemma}
\label{AHIH}
The equation \eqref{sca1} has an exponential dichotomy on $[T,+\infty)$  (\textit{resp}. $(-\infty,-T]$) with $T>0$ if and only
if it has an exponential dichotomy on $[0,+\infty)$ (\textit{resp}. $(-\infty,0]$).
\end{lemma}

\begin{proof}
See Appendix A.
\end{proof}

A noteworthy fact is that a linear equation (\ref{sca1}) can have an exponential dichotomy simultaneously on $(-\infty,0]$
and $[0,+\infty)$ but does not necessarily has an exponential dichotomy on $(-\infty,+\infty)$. For example, consider $\mathfrak{a}\colon \mathbb{R}\to \mathbb{R}$ defined by
\begin{displaymath}
\mathfrak{a}(t)=\left\{\begin{array}{ccl}
1 &\textnormal{if} &  t\leq -T \\
\phi(t)  &\textnormal{if} & t\in [-T,T] \\
-1 &\textnormal{if} &  t\geq T, 
\end{array}\right.
\end{displaymath}
where $T>0$, $\phi\colon [-T,T]\to \mathbb{R}$ is continuous such that $\phi(-T)=1$ and $\phi(T)=-1$.
In fact, it is easy to see that (\ref{DEPI}) is verified on $[T,+\infty)$ 
with $K=1$ and $\alpha=1$, while  (\ref{DENP}) is verified on $(-\infty,-T]$ with the same constants. Now, by Lemma \ref{AHIH}
we have that $x'=a(t)x$ has an exponential dichotomy on $(-\infty,0]$ and $[0,+\infty)$. 
Nevertheless, by using Remark \ref{BST} we have that this equation has not an exponential dichotomy on $(-\infty,+\infty)$ because any solution is bounded on $\mathbb{R}$. Finally, we stress that (\ref{sca1}) has an exponential dichotomy on $(-\infty,+\infty)$
if and only if it also has an exponential dichotomy on $(-\infty,0]$
and $[0,+\infty)$ with the same projector.

\subsection{Roughness, admissibility and spectra}

The next result is known as the \textit{roughness property} and shows that the set of bounded continuous functions such that an exponential dichotomy exists is open. More precisely,
\begin{proposition}
\label{rough-a}
Assume that  \eqref{sca1} has an exponential  
dichotomy on $[0,+\infty)$ with constants $C$, $\alpha$
and let
$\mathfrak{e}\colon [0,+\infty)\to \mathbb{R}$ be a continuous function such that
\begin{displaymath}
\delta_{0}:= \limsup\limits_{t\to +\infty}|\mathfrak{e}(t)|<{\alpha}.
\end{displaymath}
Then the perturbed equation
\begin{displaymath}
\dot{x}=[\mathfrak{a}(t)+\mathfrak{e}(t)]x    
\end{displaymath}
also has an exponential dichotomy on $[0,+\infty)$ with constants $\tilde{C},\alpha_{0}$ where
$\tilde{C}>C$ and $\alpha_{0}$ is any constant satisfying the inequality
$\alpha_0+\delta_{0} < \alpha$.
\end{proposition}

\begin{proof}
See Appendix A.
\end{proof}

Now, let $(\mathcal{B},\mathcal{D})$ be a pair of function spaces,
this pair is said
to be \textit{admissible} if for any $v\in \mathcal{D}$, the equation (\ref{sca0}) has --at least-- one solution in the space $\mathcal{B}$. The next result describes the equivalence between the exponential dichotomy on $[0,+\infty)$ with the admissibility property 
when $\mathcal{B}=\mathcal{D}$ is the Banach space of bounded continuous functions on 
$[0,+\infty)$ with the supremum norm:

\begin{proposition} 
\label{Prop-Cop}
Let $t\mapsto \mathfrak{a}(t)$ be a continuous function on $[0,+\infty)$. The following properties are equivalent:
\begin{itemize}
\item[1.-] The equation \eqref{sca1} has an exponential dichotomy on $[0,+\infty)$.
\item[2.-] For any 
bounded and continuous function $\mathfrak{v}\colon [0,+\infty) \to \mathbb{R}$, the problem
\begin{equation}
\label{inhp}
x'=\mathfrak{a}(t)x + \mathfrak{v}(t)    
\end{equation}
has 
at least one bounded and continuous solution on $[0,+\infty)$. 
More precisely, \eqref{proj-id} holds if and only if 
all solutions are bounded for each bounded $\mathfrak{v}(\cdot)$   and
\eqref{proj-null} holds if and only if 
there exists exactly one bounded solution for  each bounded $\mathfrak{v}(\cdot)$.
\item[3.-] The non homogeneous equation 
\begin{equation}
\label{3nh}
z'=\mathfrak{a}(t)z+ 1
\end{equation}
has at least a bounded solution on $[0,+\infty)$. 
\item[4.-] There exists a bounded and continuous function $\mathfrak{v}\colon [0,+\infty)\to \mathbb{R}$ with 
$$
\liminf\limits_{t\to +\infty}|\mathfrak{v}(t)|>0
$$
such that the problem \eqref{inhp} has at least one bounded solution. 
\end{itemize}
\end{proposition}

\begin{proof}
See Appendix A.
\end{proof}

We stress the remarkable simplicity of the proof of the above lemma
compared with the 
general $n$--dimensional case which has been addressed from a functional analysis approach in
\cite{Coppel} and variational methods in \cite{Campos}.

\medskip

The exponential dichotomy spectrum has been introduced
by R. Sacker and G. Sell in \cite{SS-1978} in a skew--product framework and revisited
in a simpler form by B. Aulbach and S. Siegmund in \cite{Aulbach, Siegmund-2002}:
\begin{definition}
\label{DEFSPEC}
The exponential dichotomy spectrum of \eqref{sca1} is the set:
\begin{displaymath}
\Sigma_{J}(\mathfrak{a}):=\left\{\lambda \in \mathbb{R} \colon  \dot{x}=[\mathfrak{a}(t)-\lambda]x \quad \textnormal{has not an exponential dichotomy on $J$}\right\}.   
\end{displaymath}
\end{definition}

The above spectrum has a simpler characterization in the scalar case:
\begin{proposition} 
\label{SPEC}
The exponential dichotomy spectrum of \eqref{sca1} is characterized as follows:
\begin{displaymath}
\Sigma_{[0,+\infty)}(\mathfrak{a})=[\beta^{-}(\mathfrak{a}),\beta^{+}(\mathfrak{a})] 
\end{displaymath}
where $\beta^{-}(\mathfrak{a})$ and $\beta^{+}(\mathfrak{a})$ are the lower and upper Bohl exponents, which are 
respectively defined by:
\begin{displaymath}
\beta^{-}(\mathfrak{a}):=\liminf\limits_{L,s\to +\infty}\frac{1}{L}\int_{s}^{s+L}\mathfrak{a}(r)\,dr\quad \textnormal{and} \quad \beta^{+}(\mathfrak{a}):=\limsup\limits_{L,s\to +\infty}\frac{1}{L}\int_{s}^{s+L}\mathfrak{a}(r)\,dr. 
\end{displaymath}
More precisely, 
\begin{itemize}
    \item The case \eqref{proj-id} holds for $\mathfrak{a}(\cdot)-\lambda$ if and only if $\lambda < \beta^{-}(\mathfrak{a})$.
    \item The case \eqref{proj-null} holds for $\mathfrak{a}(\cdot)-\lambda$ if and only if $\lambda > \beta^{+}(\mathfrak{a})$.
\end{itemize}
\end{proposition}

\begin{proof}
See Appendix A.
\end{proof} 

We refer the reader to \cite{barabanov2001} and \cite{Dalecki} for more details about the Bohl exponents.

\subsection{The almost periodic case}
When $\mathfrak{a}(\cdot)$ is almost periodic, it is well known that $\Sigma_{\mathbb{R}}(\mathfrak{a})=\Sigma_{[0,+\infty)}(\mathfrak{a})$ and we refer the reader to \cite[Remark 5.9]{Kloeden} for details. In addition, the average (\ref{average-a}) coincides with the upper and lower Bohl exponents
which, in turn, implies that the exponential dichotomy spectrum is the singleton:
$$
\Sigma_{\mathbb{R}}(\mathfrak{a})=\{\mathcal{M}[\mathfrak{a}]\}.
$$

The above identity combined with $\lambda=0$ in Definition \ref{DEFSPEC} implies directly that:
\begin{proposition}
\label{AVED}
Let $\mathfrak{a}\colon \mathbb{R}\to \mathbb{R}$ be continuous and almost periodic. 
Then \eqref{sca1} has an exponential dichotomy if and only if
 $\mathcal{M}[\mathfrak{a}]\neq 0$.
 In particular: 
\begin{itemize}
\item[a)] The property \eqref{proj-id} is verified on $\mathbb{R}$ if and only if $\mathcal{M}[\mathfrak{a}]<0$,
\item[b)] The property \eqref{proj-null} is verified on $\mathbb{R}$ if and only if  $\mathcal{M}[\mathfrak{a}]>0$.
\end{itemize}
\end{proposition}

The next result is about almost periodic admissibility. If the linear almost periodic equation (\ref{sca1})
has an exponential dichotomy on $\mathbb{R}$ and $\mathfrak{v}(\cdot)$ is almost periodic then the non homogeneous equation (\ref{sca0}) has a unique almost periodic solution a we refer to \cite{Fink} for a proof.
\begin{proposition}
\label{AVAD}
If $\mathfrak{a}\colon \mathbb{R}\to \mathbb{R}$ is a continuous and almost periodic function
with $\mathcal{M}[\mathfrak{a}]\neq 0$ then, for any almost periodic function $t\mapsto \mathfrak{v}(t)$,
the non homogeneous  equation \eqref{sca0} has a unique almost periodic solution $t\mapsto y^{*}(t)$. In particular: 
\begin{itemize}
\item[a)] When $\mathcal{M}[\mathfrak{a}]<0$, the solution $t\mapsto y^{*}(t)$ is given by
$$
y^{*}(t)=\int_{-\infty}^{t}e^{\int_{s}^{t}\mathfrak{a}(r)\,dr}\mathfrak{v}(s)\,ds,
$$
\item[b)] When $\mathcal{M}[\mathfrak{a}]>0$, the solution $t\mapsto y^{*}(t)$ is given by
$$
y^{*}(t)=-\int_{t}^{\infty}e^{\int_{s}^{t}\mathfrak{a}(r)\,dr}\mathfrak{v}(s)\,ds.
$$
\end{itemize}
\end{proposition}

 A direct consequence of Propositions \ref{AVED} and \ref{AVAD} is the following
result:
\begin{corollary}
\label{cor1}
\label{atractivite}
Let us consider the non homogeneous equation \eqref{sca0}, where
$\mathfrak{a}(\cdot)$ and $\mathfrak{v}(\cdot)$ are continuous and almost periodic functions such that $\mathcal{M}[\mathfrak{a}]\neq 0$. Then: 
\begin{itemize}
\item[a)] If $\mathcal{M}[\mathfrak{a}]<0$, then there exists $K\geq 1$ and $\alpha>0$ such that any solution $t\mapsto y(t)$ of \eqref{sca0}
verifies
$$
|y(t)-y^{*}(t)|\leq Ke^{-\alpha(t-t_{0})}|y(t_0)-y^{*}(t_{0})| \quad \textnormal{for any 
$t\geq t_0$}.
$$
\item[b)] If $\mathcal{M}[\mathfrak{a}]>0$, then there exists $K\geq 1$ and $\alpha>0$ such that any solution $t\mapsto y(t)$ of \eqref{sca0}
verifies
$$
|y(t)-y^{*}(t)|\leq Ke^{-\alpha(t_{0}-t)}|y(t_0)-y^{*}(t_{0})| \quad \textnormal{for any 
$t\leq t_0$}.
$$
\end{itemize}
\end{corollary}

\section{The dynamics of the superior competitor ($n=1$)}
When $n=1$, the system (\ref{tilman1-TV}) becomes the scalar differential equation
\begin{equation}
\label{EEN1}
p_{1}'=c_1(t) p_1\left(h-p_{1} \right)-m_1(t)p_1,
\end{equation}
which is the Levins model (\ref{levins-0}) with almost periodic colonization/extinction rates and habitat destruction. The study of this 
equation will be addressed by using the change of variables $u=1/p_{1}$, which transforms (\ref{EEN1}) into 
    \begin{equation}
    \label{lineaire}
    u'= a_{11}(t)u + c_{1}(t),
    \end{equation}
    where the almost periodic function $t\mapsto a_{11}(t)$ is defined by 
\begin{equation}
\label{variables-aux}
 a_{11}(t)= -[c_1(t)h-m_1(t)]. 
\end{equation} 

\subsection{Preliminiaries and main result}
In this section, we will see that the strong persistence and/or strong extinction of the single species metapopulation described by (\ref{EEN1}) is determined by the properties of the linear part of (\ref{lineaire}), namely, the differential equation
\begin{equation}
    \label{lineaire-hom}
    v'= a_{11}(t)v.
    \end{equation}

From an ecological point of view, the almost periodic equation (\ref{lineaire-hom}) describes the balance between the opposing forces of colonization and extinction. The predominance of one of these forces is determined by the sign of the average $\mathcal{M}[a_{11}]$:
\begin{itemize}
\item[$\bullet$] The property $\mathcal{M}[a_{11}]<0$ means that $h\mathcal{M}[c_{1}]>\mathcal{M}[m_{1}]$, that is, the colonization term $t\mapsto h c_{1}(t)$ is dominant --in average-- over the extinction rate $m_{1}(t)$. Hence, the strong persistence of the metapopulation is expected. 
\item[$\bullet$] The property $\mathcal{M}[a_{11}]>0$ means that $h\mathcal{M}[c_{1}]<\mathcal{M}[m_{1}]$, that is, the extinction rate dominates --in average-- over the colonization $t\mapsto h c_{1}(t)$ and the strong extinction of the metapopulation is expected.
\end{itemize}

The above results also can be interpreted from an habitat destruction perspective by considering
the \textit{habitability threshold}:
\begin{equation}
\mathcal{R}_{0}=\frac{\mathcal{M}[m_{1}]}{\mathcal{M}[c_{1}]},
\end{equation}
which has been used by R.S. Etienne in \cite{Etienne} and P. Marquet in \cite{Marquet}. We will see that $\mathcal{R}_{0}$ provides a lower bound 
of habitable patches ensuring the persistence
of the metapopulation. In fact, it is straightforward to verify that:
\begin{equation}
\label{resumen1}
\mathcal{M}[a_{11}]<0 \iff h\mathcal{M}[c_1] >\mathcal{M}[m_1] \iff \mathcal{R}_{0}<h,
\end{equation}
and 
\begin{equation}
\label{resumen2}
\mathcal{M}[a_{11}]>0 \iff h\mathcal{M}[c_1] <\mathcal{M}[m_1] \iff h<\mathcal{R}_{0}.
\end{equation}

The equation (\ref{resumen1}) means that the predominance of colonization over the extinction forces is equivalent to a fraction $h$ of habitable patches larger than the habitability threshold $\mathcal{R}_{0}$. Analogously, (\ref{resumen2}) means that subordination of the colonization forces towards the extinction ones is equivalent to a fraction $h$ of habitable patches inferior than the habitability threshold.

The next result will state that the strong persistence of the metapopulation is described by the existence of an almost periodic non autonomous equilibrium $t\mapsto p_{1,1}^{\ast}(t)$ which is exponentially attractive while the strong extinction
is described by the exponential stability of the trivial solution $p(t)=0$. In addition, both possible scenarios are determined by the habitability threshold
$\mathcal{R}_{0}$. The proof is based on the fact that the above mentioned equivalences (\ref{resumen1})--(\ref{resumen2}) are also equivalent to the exponential dichotomy of the linear equation (\ref{lineaire-hom}).

\begin{theorem}
\label{Prop-1}
If the colonization and extinction rates are such that $\mathcal{R}_{0}<h$,
then the fraction of occupied patches by the species 1 at time $t$, namely $t\mapsto p_{1}(t)$,
is strongly persistent and verifies an exponential convergence:
\begin{equation}
\label{sol1}
\lim\limits_{t\to +\infty}|p_{1}(t)-p_{1,1}^{*}(t)|=0,
\end{equation}
where
$t\mapsto p_{1,1}^{*}(t)$ is the almost periodic function:
\begin{equation}
\label{sol-periodique}
p_{1,1}^{*}(t)=\left(\int_{-\infty}^{t}c_{1}(s) e^{\int_s^t a_{11}(r)\,dr}\,ds
    \right)^{-1}.
\end{equation}
On the other hand, if the colonization and extinction rates are such that 
$h< \mathcal{R}_{0}$,
then the fraction of occupied patches decays exponentially towards zero when $t\to +\infty$.
\end{theorem}

\begin{proof}
As we know by (\ref{resumen1}) that $\mathcal{R}_{0}<h$ is equivalent to $\mathcal{M}[a_{11}]<0$, the statement a) of Proposition
\ref{AVED} says that the linear equation (\ref{lineaire-hom}) has an exponential dichotomy on the real line
with the identity projector and constants $K\geq 1$ and $\alpha>0$. Moreover, the statement a) of Proposition
\ref{AVAD} says that the inhomogenous  equation (\ref{lineaire}) has a unique almost periodic 
solution $t\mapsto u^{*}(t)$ described by 
\begin{displaymath}
u^{*}(t)=\int_{-\infty}^{t}c_{1}(s) e^{\int_s^t a_{11}(r)\,dr}\,ds,
\end{displaymath}
while the Corollary \ref{cor1} says that any other solution $t\mapsto u(t)$ of (\ref{lineaire}) verifies
$$
|u(t)-u^{*}(t)|\leq Ke^{-\alpha (t-t_{0})}|u(t_{0})-u^{*}(t_{0})| \quad \textnormal{for any $t\geq t_{0}$}.
$$

By using (\ref{infap}) and the positiveness of $m_{1}(\cdot)$ we can prove that 
$$
\inf\limits_{-\infty <t<+\infty}u^{*}(t)\geq \frac{c_{1}^{-}}{h||c_{1}||_{\infty}},
$$
and by using Theorem 7 from \cite{Besicovitch} it follows that $t\mapsto p_{1,1}^{*}(t)=1/u^{*}(t)$
is almost periodic. Moreover, as $p_1(t)=1/u(t)\leq h$, this implies that
\begin{displaymath}
\begin{array}{rcl}
|p_{1}(t)-p_{1,1}^{*}(t)|&=& \displaystyle \left| \frac{1}{u(t)}-\frac{1}{u^{*}(t)}\right|\\\\
&=& \displaystyle p_1(t)p_{1,1}^{*}(t)|u(t)-u^{*}(t)|\\\\
&\leq& \displaystyle Ke^{-\alpha (t-t_{0})}|u(t_{0})-u^{*}(t_{0})|,
\end{array}
\end{displaymath}
then we can deduce
\begin{displaymath}
|p_{1}(t)-p_{1,1}^{*}(t)|\leq \frac{K}{p_{1}(t_{0})p_{1,1}^{*}(t_{0})}
e^{-\alpha(t-t_{0})}|p_{1}(t_{0})-p_{1,1}^{*}(t_{0})|   
\end{displaymath}
and the exponential convergence follows.

{We point out that $(\ref{sol-periodique})$ does not provide information
about whether or not $p_{1,1}^{*}(0)\in (0,h)$ is satisfied. Nevertheless, let us recall that any non-constant almost periodic function has infinite growth and decay intervals since, otherwise, the function will have a limit, which is not possible. By using this fact combined with the differentiability of $p_{1,1}^{*}(\cdot)$, we can see that, at any local maximum $t_{0}$, the function $p_{1,1}^{*}(\cdot)$ verifies
$$
p_{1,1}^{*}(t_{0})=h-\frac{m_1(t_{0})}{c_{1}(t_{0})}<h,
$$
which allows us to deduce that $p_{1,1}^{*}(t)\in (0,h)$ for any $t\in \mathbb{R}$.}


Now, by (\ref{resumen2}) we know that that the strict inequality $h<\mathcal{R}_{0}$ 
is equivalent to $\mathcal{M}[a_{11}]>0$. The statement b) of the Proposition
\ref{AVED} says that the linear equation
(\ref{lineaire-hom}) has an exponential dichotomy on the real line with null projector, that is, there exist $K\geq 1$ and $\alpha>0$
such that
\begin{displaymath}
e^{\int_{s}^{r}a_{11}(\tau)\,d\tau}\leq Ke^{-\alpha(s-r)} \quad \textnormal{for any $r< s$},     \end{displaymath}
then any solution $t\mapsto v(t)$ of (\ref{lineaire-hom}) verifies
$$
K^{-1}e^{\alpha (t-t_{0})}|v(t_{0})|\leq |v(t)|=e^{\int_{t_{0}}^{t}a_{1}(\tau)\,d\tau}|v(t_{0})| \quad \textnormal{for any $t\geq t_{0}$}.
$$

The statement 2.- from Proposition \ref{Prop-Cop} says that
the equation (\ref{lineaire}) has a unique solution bounded on $[0,+\infty)$. A careful reading of the proof, combined with the positiveness of $c_{1}(\cdot)$ shows that this bounded solution
has a negative initial condition. In consequence, the solutions of (\ref{lineaire}) with $u(t_{0})=u_{0}\geq 1$ are unbounded. We will explore in depth this fact:
\begin{displaymath}
\begin{array}{rcl}
    u(t) & = & \displaystyle u_0 e^{\int_{t_0}^t a_{11}(s)ds} + \int_{t_0}^t c_1(s) e^{\int_s^t a_{11}(r)\,dr}\,ds \\\\
     & \geq  & \displaystyle  K^{-1}e^{\alpha (t-t_{0})}+ \int_{t_0}^t c_1(s) e^{\int_s^t a_{11}(r)\,dr}\,ds \\\\
     &\geq & \displaystyle  K^{-1}e^{\alpha (t-t_{0})},
\end{array}    
\end{displaymath}    
now, as $p_{1}(t)=1/u(t)$, we can deduce that $0\leq p(t) \leq  Ke^{-\alpha(t-t_{0})}$
and the extinction of the metapopulation at exponential rate follows.

\end{proof}

In addition, if $\mathcal{R}_{0}<h$ is verified, then it can be deduced from (\ref{EEN1}) and (\ref{resumen1}) that
the average of the almost periodic function $t\mapsto c_{1}(t)p_{1,1}^{*}(t)$ is:
\begin{equation}
\label{CAPN1}
\mathcal{M}[c_{1}p_{1,1}^{*}]=h\mathcal{M}[c_{1}]-\mathcal{M}[m_{1}]=\mathcal{M}[c_{1}](h-\mathcal{R}_{0}).
\end{equation}

\subsection{Habitability thresholds for the $i$--th species}

Let $e_{i}$ be $i$--th component of the canonical basis of $\mathbb{R}^{n}$ and $\delta\in (0,h)$. Notice that
any solution of (\ref{tilman1-TV}) with initial condition $p(0)=\delta e_{i}$ has a unique positive component which is the solution of
the equation
\begin{displaymath}
p_{i}'=c_i(t) p_i\left(h-p_{i} \right)-m_i(t)p_i,
\end{displaymath}
and describes the fraction of the occupied patches by the $i$--th species in absence of competitors.

A direct consequence of Theorem \ref{Prop-1} is that, for any initial condition
$\delta e_{i}$ with $\delta\in (0,h)$, the strong persistence or strong extinction of the $i$--th species is determined either by $h>\mathcal{R}_{i-1}$
or $\mathcal{R}_{i-1}> h$, where $\mathcal{R}_{i-1}$ is the \textit{natural habitability threshold for the $i$--species}, which is defined by:
\begin{equation}
\label{GTH}
\mathcal{R}_{i-1}=\frac{\mathcal{M}[m_{i}]}{\mathcal{M}[c_{i}]}  \quad \textnormal{for any $i=1,\ldots,n$}.
\end{equation}

In case that the averages extinction rates verify:
\begin{equation}
\label{EAM}
\mathcal{M}[m_{1}]=\ldots=\mathcal{M}[m_{n}]=m,   
\end{equation}
then the competition--colonization trade--off (\ref{BCWC}) and (\ref{GTH}) will imply the i\-ne\-qua\-lity
\begin{displaymath}
    \mathcal{R}_{n-1}<\mathcal{R}_{n-2}<\ldots <\mathcal{R}_{1}<\mathcal{R}_{0},
\end{displaymath}
which means that, if the extinction rates have similar average, then the less advantaged (in average) competitors is the species that requires the lowest natural habitability threshold. 

In the next sections, we will see that a consequence of the assumption (\ref{BCWC}) is that lower habitability thresholds 
tend to favor the persistence of weaker competitors as well as to harm the persistence of stronger competitors even when
(\ref{EAM}) is not verified.

\section{The planar case $n=2$}
For $n=2$, the metapopulation model described by (\ref{tilman1-TV}) becomes 
\begin{equation}
\label{tilman1-TV-n2}
\left\{
\begin{array}{rcl}
\displaystyle p_{1}' &=&  p_{1}\left\{c_1(t)\left(h-p_1\right)-m_1(t)\right\} \\
\displaystyle p_{2}' &=& p_{2}\left\{c_2(t)\left(h-p_{1}-p_{2} \right)-m_2(t)-c_1(t) p_1 \right\},
\end{array}\right.
\end{equation}
with initial conditions
\begin{equation}
\label{TH}
0<p_{1}(t_{0})+p_{2}(t_{0})<h.
\end{equation}

The possible scenarios of strong persistence and/or strong extinction for the system (\ref{tilman1-TV-n2}) are described by the global
exponential attractivity of four almost periodic nonautonomous equilibria, which are summarized in the following table:
\begin{table}[H]
    \centering
    \begin{tabular}{|c|c|c|}
\hline    
                              & Persistence species 1 ($\mathcal{R}_{0}< h$)    &  Extinction species 1   ($h<\mathcal{R}_{0}$)\\
\hline                                 
Persistence species 2          & \small{Global persistence}  &     \small{Fugitive species persists} \\ 
 & \footnotesize{ $(p_{1,1}^{\ast}(t),p_{2,1}^{\ast}(t))$}  &     \footnotesize{$(0,p_{2,3}^{\ast}(t))$}  \\
\hline 
Extinction species 2          & \small{Advantaged competitor persists}     &   \small{Global extinction} \\
&  \footnotesize{$(p_{1,1}^{\ast}(t),p_{2,2}^{\ast}(t))$ with $p_{2,2}^{\ast}(t)=0$}     &  \footnotesize{$(0,p_{2,4}^{\ast}(t))$ with $p_{2,4}^{\ast}(t)=0$} \\
\hline
    \end{tabular}
    \caption{The Strong Persistence/Strong Extinction scenarios}
    \label{Table1}
\end{table}

The almost periodic function $t\mapsto p_{1,1}^{\ast}(t)$ is described by (\ref{sol-periodique}) whereas the other
ones $p_{2,j}^{\ast}(\cdot)$ will be defined throughout this section. We stress that the double index $2,j$ refers to the second species 
in competition and the $j$--th scenario with $j\in \{1,2,3,4\}$.

\subsection{Technical results}
The following technical results will be useful for our study of the
metapopulation model (\ref{tilman1-TV-n2})--(\ref{TH}).

\begin{lemma}
\label{LT1}
Let $\mathfrak{a}\colon \mathbb{R}\to \mathbb{R}$ be almost periodic with $\mathcal{M}[\mathfrak{a}]<0$ and consider two bounded continuous functions $\mathfrak{b}\colon [0,+\infty)\to (0,+\infty)$ and $\psi\colon [0,+\infty)\to \mathbb{R}$ where $\psi(t)=\textit{o}(1)$, then 
there exists $\theta>0$ such that any 
solution of
\begin{equation}
\label{inhomogeneous}
u'=[\mathfrak{a}(t)+\psi(t)]u+\mathfrak{b}(t)  \quad \textnormal{with $t\geq 0$}
\end{equation}
verifies $\limsup\limits_{t\to +\infty}|u(t)|\leq \theta$.
\end{lemma}

\begin{proof}
As $\mathcal{M}[\mathfrak{a}]<0$, Proposition \ref{AVED} says that the differential equation $x'=\mathfrak{a}(t)x$
has an exponential dichotomy on $\mathbb{R}$ with the identity projector. Since this last equation also has an exponential dichotomy on $[0,+\infty)$ with the same projector and
$\psi(t)=\textit{o}(1)$, Proposition \ref{rough-a} says that $z'=[\mathfrak{a}(t)+\psi(t)]z$ also has an exponential dicothomy on $[0,+\infty)$ with the identity projector, that is, there exist $K\geq 1$ and $\alpha>0$ such that
$$
e^{\int_{t_{0}}^{t}[\mathfrak{a}(s)+\psi(s)]\,ds} \leq K e^{-\alpha(t-t_{0})} \quad \textnormal{for any $t\geq t_{0}\geq 0$}.
$$

Let $\mathfrak{b}_{\infty}:=\limsup\limits_{t\to +\infty}\mathfrak{b}(t)$ and note that for any $\varepsilon>0$ there exists
$t_{0}$ large enough such that:
\begin{displaymath}
\begin{array}{rcl}
|u(t)| &\leq & Ke^{-\alpha(t-t_{0})}|u(t_{0})|+ K(||\mathfrak{b}||_{\infty}+\varepsilon)\int_{t_{0}}^{t}e^{-\alpha(t-s)}\,ds \\\\
&\leq&  Ke^{-\alpha(t-t_{0})}|u(t_{0})|+\frac{K}{\alpha}(||\mathfrak{b}||_{\infty}+\varepsilon)(1-e^{-\alpha(t-t_{0})}).
\end{array}
\end{displaymath}

Finally, the result follows with $\theta:=(K||\mathfrak{b}||_{\infty}+\varepsilon)/\alpha$ by taking the upper limit.

\end{proof}

\begin{lemma}
\label{LT2}
Let us consider the equation
\begin{equation}
\label{inho-n2}
v'=\mathfrak{a}(t)v+\mathfrak{b}(t)  \quad \textnormal{with $t\geq 0$},
\end{equation}
where $\mathfrak{a}\colon \mathbb{R}\to \mathbb{R}$ is almost periodic with $\mathcal{M}[\mathfrak{a}]<0$ and
$\mathfrak{b}\colon [0,+\infty)\to \mathbb{R}$ is a continuous function. If $\mathfrak{b}(t)=\textit{o}(1)$,
then $v(t)=\textit{o}(1)$. Moreover, if $\mathfrak{b}(\cdot)$ converges exponentially, then $v(\cdot)$ also converges exponentially. 
\end{lemma}

\begin{proof}
As $\mathcal{M}[\mathfrak{a}]<0$, Proposition \ref{AVED} says that the differential equation $x'=\mathfrak{a}(t)x$
has an exponential dichotomy on $\mathbb{R}$ and $[0,+\infty)$ with the identity projector, that is, there exists $K\geq 1$ and $\alpha>0$ such that
$$
e^{\int_{t_{0}}^{t}\mathfrak{a}(s)\,ds} \leq K e^{-\alpha(t-t_{0})} \quad \textnormal{for any $t\geq t_{0}\geq 0$}.
$$
Firstly, 
assume that $\mathfrak{b}(\cdot)$ decays exponentially, that is, there exists $\mathcal{K}\colon [0,+\infty)\to [0,+\infty)$ and $\gamma>0$ such that:
$$
|\mathfrak{b}(t)|\leq \mathcal{K}(t_{0})e^{-\gamma(t-t_{0})} \quad \textnormal{for any $t\geq t_{0}\geq 0$},
$$
and, without loss of generality, it can be assumed that $\alpha \neq \gamma$.

Now, let us note that any solution of (\ref{inho-n2}) verifies
\begin{displaymath}
\begin{array}{rcl}
|v(t)| &\leq & Ke^{-\alpha(t-t_{0})}|v(t_{0})|+ K\mathcal{K}(t_{0})\int_{t_{0}}^{t}e^{-\alpha(t-s)}e^{-\gamma(s-t_{0})}\,ds \\\\
&\leq&  Ke^{-\alpha(t-t_{0})}|v(t_{0})|+K\mathcal{K}(t_{0})e^{-\alpha t}e^{\gamma t_{0}}\int_{t_{0}}^{t}e^{(\alpha-\gamma)s}\,ds \\\\
&\leq &  Ke^{-\alpha(t-t_{0})}|v(t_{0})|+\frac{K\mathcal{K}(t_{0})}{\alpha-\gamma}e^{-\alpha t}e^{\gamma t_{0}}
\left\{e^{(\alpha-\gamma)t}-e^{(\alpha-\gamma)t_{0}}\right\} \\\\
&\leq &  Ke^{-\alpha(t-t_{0})}|v(t_{0})|+\frac{K\mathcal{K}(t_{0})}{|\alpha-\gamma|}\left(e^{-\gamma(t-t_{0})}-e^{-\alpha(t-t_{0})}\right).
\end{array}
\end{displaymath}

Finally, we will prove the general case $b(t)=\textit{o}(1)$: As $\mathcal{M}[\mathfrak{a}]<0$, it follows by Proposition \ref{AVED}
that the inequality (\ref{proj-id}) is verified with constants $C\geq 0$ and $\alpha>0$. Then by choosing $T>C/\alpha$ large enough, for any $t-s>T$, we have:
$$
\int_{s}^{t}\mathfrak{a}(r)\,dr\leq C-\alpha(t-s)\leq \frac{C}{T}(t-s)-\alpha(t-s),
$$
and there exists $c>\alpha-(C/T)>0$ such that $\int_s^t \mathfrak{a}(\tau)\,d\tau < -c(t-s) <0 $ whenever $t-s>T$. Next, given $\varepsilon>0$ we may fix $t_0$ such that $|\mathfrak{b}(t)|<\tilde\varepsilon$ for $t>t_0$, with $\tilde\varepsilon$ to be specified. Note that any solution of (\ref{inho-n2}) can be written as follows
$$
v(t) =v(t_0)e^{\int_{t_0}^{t}\mathfrak{a}(\tau)\,d\tau} +\int_{t_0}^t \mathfrak{b}(s)e^{\int_{s}^{t}\mathfrak{a}(\tau)\,d\tau}\,ds,
$$
and we can deduce that
$$
\left|v(t) -v(t_0)e^{\int_{t_0}^{t}\mathfrak{a}(\tau)\,d\tau}\right| \le \tilde\varepsilon\int_{t_0}^t e^{\int_{s}^{t}\mathfrak{a}(\tau)\,d\tau}ds.
$$
Next, for $t$ large enough we may write
\begin{displaymath}
\begin{array}{rcl}
\displaystyle\int_{t_0}^te^{\int_{s}^{t}\mathfrak{a}(\tau)\,d\tau}\,ds &=&
\displaystyle \int_{t_0}^{t-T} e^{\int_{s}^{t}\mathfrak{a}(\tau)\,d\tau}\,ds +
\int_{t-T}^{t} e^{\int_{s}^{t}\mathfrak{a}(\tau)\,d\tau}ds \\\\
&\le & \displaystyle\int_{t_0}^{t-T} e^{-c(t-s)}ds +
\int_{t-T}^{t} e^{\|\mathfrak{a}\|_\infty T}ds \\\\
&\leq &\frac 1c + Te^{\|\mathfrak{a}\|_\infty T}.
\end{array}
\end{displaymath}

Setting $\tilde\varepsilon$ such that 
$$\tilde\varepsilon\left(\frac 1c + Te^{\|\mathfrak{a}\|_\infty T}
\right)\le \varepsilon$$
and passing to the limit and recalling that $e^{\int_{t_0}^{t}\mathfrak{a}(\tau)\,d\tau}=\textit{o}(1)$, it is seen that 
$\limsup\limits_{t\to\infty} |v(t)|\le \varepsilon$ and the Lemma follows.
\end{proof}

\begin{lemma}
\label{LT11}
Let $\mathfrak{a}\colon \mathbb{R}\to \mathbb{R}$ be almost periodic with $\mathcal{M}[\mathfrak{a}]>0$ and consider two bounded continuous functions $\mathfrak{c}\colon [0,+\infty)\to (0,+\infty)$ and $\phi\colon [0,+\infty)\to \mathbb{R}$ where $\phi(t)=\textit{o}(1)$, then 
there exists constants $K>0$ and $\alpha_{0}>0$ such that any 
solution $t\mapsto v(t)$ of
\begin{equation}
\label{inhomogeneous-22}
v'=[\mathfrak{a}(t)+\phi(t)]v+\mathfrak{c}(t)  \quad \textnormal{with $t\geq 0$}
\end{equation}
with $v(t_{0})\geq 1$ verifies $v(t)\geq K^{-1}e^{\alpha_{0}(t-t_{0})}$.
\end{lemma}

\begin{proof}
As $\mathfrak{a}(\cdot)$ is almost periodic and $\mathcal{M}[\mathfrak{a}]>0$, the statement b) of Proposition \ref{AVED}  allows us to deduce  
the existence of $C\geq 1$ and $\alpha>0$ such that
$$
e^{\int_{t_{0}}^{t}\mathfrak{a}(\tau)\,d\tau} \geq C^{-1}e^{\alpha(t-t_{0})} \quad \textnormal{for any $t\geq t_0$}.
$$

In addition, as $\phi(t)=\textit{o}(1)$, the Proposition \ref{rough-a} allow us to deduce the existence of
$K\geq 1$ and $\alpha_{0}\in (0,\alpha)$ such that
$$
e^{\int_{t_{0}}^{t}[\mathfrak{a}(s)+\phi(s)]}\geq K^{-1}e^{\alpha_{0}(t-t_{0})} \quad \textnormal{for any $t>t_{0}\geq 0$}.
$$

The above estimation, combined with the positiveness of $\mathfrak{c}(\cdot)$ implies that, for any $t>t_{0}\geq 0$, the solutions of (\ref{inhomogeneous-22}) with $v(t_{0})=v_{0}\geq 1$ verify:
\begin{displaymath}
\begin{array}{rcl}
    v(t) & = & \displaystyle v_0 e^{\int_{t_0}^t [\mathfrak{a}(s)+\phi(s)]\,ds} + \int_{t_0}^t \mathfrak{c}(s) e^{\int_s^t [\mathfrak{a}(r)+\phi(r)]\,dr}\,ds \\\\
     & \geq  & \displaystyle  K^{-1}e^{\alpha_{0} (t-t_{0})}+ \int_{t_0}^t \mathfrak{c}(s) e^{\int_s^t [\mathfrak{a}(r)+\phi(r)]\,dr}\,ds \\\\
     &\geq & \displaystyle  K^{-1}e^{\alpha_{0}(t-t_{0})},
\end{array}    
\end{displaymath}  
and the Lemma follows.
\end{proof}

\subsection{Strong persistence/extinction scenarios when $\mathcal{R}_{0}<h$}
Now, we will find sufficient conditions for all and each one of the scenarios described in the first column of the Table \ref{Table1}. 

In order to
do that, we will see that the solutions of (\ref{tilman1-TV-n2}) will be studied with the help of the almost periodic equation
\begin{equation}
\label{a2}
\dot{z}=a_{21}(t)z \quad \textnormal{with} \quad a_{21}(t)= -[h c_2(t)-m_2(t)-\{c_1(t)+c_{2}(t)\}p_{1,1}^*(t)],
\end{equation}
which describes the balance between the colonization forces $hc_{2}(t)$ and 
extinction forces $m_{2}(t)$ with an additive positive perturbation $[c_{1}(t)+c_{2}(t)]p_{1,1}^{*}(t)$ induced by the superior competitor, where $t\mapsto p_{1,1}^{*}(t)$ is given by \eqref{sol-periodique}. Moreover, by (\ref{GTH}) and (\ref{a2}) 
combined with the positiveness of $\mathcal{M}[c_{2}]$ we can deduce that:

\begin{equation}
\label{resumen3}
\left\{\begin{array}{rcl}
\mathcal{M}[a_{21}]<0  &\iff &  h\mathcal{M}[c_2] -\mathcal{M}[m_2]-\mathcal{M}\left[(c_{1}+c_{2})p_{1,1}^*\right]>0, \\\\
&\iff & \displaystyle\mathcal{M}[c_{2}]\left(h-\mathcal{R}_{1}-\frac{\mathcal{M}[(c_{1}+c_{2})p_{1,1}^{\ast}]}{\mathcal{M}[c_{2}]}\right)>0,\\\\
&\iff & \displaystyle \mathcal{R}_{1,1}^{\ast}:=\mathcal{R}_{1}+\underbrace{\frac{\mathcal{M}[(c_{1}+c_{2})p_{1,1}^{\ast}]}{\mathcal{M}[c_{2}]}}_{:=\ell_{1,1}}<h,
\end{array}\right.
\end{equation}
and
\begin{equation}
\label{resumen4}
\left\{\begin{array}{rcl}
\mathcal{M}[a_{21}]>0  &\iff &  h\mathcal{M}[c_2] -\mathcal{M}[m_2]-\mathcal{M}\left[(c_{1}+c_{2})p_{1,1}^*\right]<0, \\\\
&\iff & \displaystyle\mathcal{M}[c_{2}]\left(h-\mathcal{R}_{1}-\frac{\mathcal{M}[(c_{1}+c_{2})p_{1,1}^{\ast}]}{\mathcal{M}[c_{2}]}\right)<0,\\\\
&\iff & \displaystyle \mathcal{R}_{1,1}^{\ast}:=\mathcal{R}_{1}+\underbrace{\frac{\mathcal{M}[(c_{1}+c_{2})p_{1,1}^{\ast}]}{\mathcal{M}[c_{2}]}}_{:=\ell_{1,1}}>h.
\end{array}\right.
\end{equation}

\begin{theorem}[Global persistence]
\label{GP}
    Assume that $\mathcal{M}[c_{1}]<\mathcal{M}[c_{2}]$ if verified. If the fraction $h$ of habitable patches together with the colonization and extinction rates are such that the inequality
    \begin{displaymath}
     \max\left\{\mathcal{R}_{0},\mathcal{R}_{1,1}^{\ast}\right\}<h, 
    \end{displaymath}
    is also verified, then the metapopulation \eqref{tilman1-TV-n2} has a unique positive almost periodic solution $t\mapsto (p_{1,1}^*(t),p_{2,1}^*(t))$, where $t\mapsto p_{2,1}^{*}(t)$ is defined by:
\begin{equation}
\label{p2}
    p_{2,1}^{*}(t)=\left(\int_{-\infty}^t e^{\int_{s}^{t}a_{21}(r)\,dr}c_2(s)\,ds\right)^{-1}. 
\end{equation}
Furthermore, for every solution $t\mapsto (p_1(t),p_2(t))$ of \eqref{tilman1-TV-n2}--\eqref{TH} it follows that 
    \begin{displaymath}
\lim\limits_{t\to +\infty}|p_{1}(t)-p_{1,1}^{*}(t)| + |p_{2}(t)-p_{2,1}^{*}(t)|=0,
\end{displaymath}
and the convergence is exponential.
\end{theorem}
 \begin{proof}
The proof will be divided into several steps.

\noindent \textit{Step 1: Existence of an almost periodic solution.} As
$\mathcal{R}_{0}<h$, by Theorem \ref{Prop-1} we have that $t\mapsto p_{1,1}^{*}(t)$ is the unique positive almost periodic solution of the first equation of (\ref{tilman1-TV-n2}). When replacing 
$t\mapsto p_{1}(t)$ by $t\mapsto p_{1,1}^*(t)$ in the second equation, we obtain:
\begin{equation}
\label{2ep}
 p_{2}'= \displaystyle p_{2}\left\{c_2(t)\left(h-p_2-p_{1,1}^{*}(t) \right)-m_2(t)-c_1(t) p_{1,1}^{*}(t) \right\}.
\end{equation}

The existence of a positive and exponentially attractive almost periodic solution of (\ref{2ep})
can be deduced analogously as in the proof of the Theorem \ref{Prop-1} by the transformation $u=1/p_{2}$, which leads to
\begin{equation}
\label{NHT3}
u'=a_{21}(t)u+c_{2}(t).
\end{equation}
Indeed, we know by (\ref{resumen3}) that $\mathcal{R}_{0}<h$ is equivalent to $\mathcal{M}[a_{21}]<0$ and Proposition \ref{AVED} implies that the equation (\ref{a2}) has an exponential dichotomy on $\mathbb{R}$ with the identity projector. That is,
there exist $K\geq 1$ and $\alpha>0$ such that
\begin{equation}
\label{DET3}
e^{\int_{t_{0}}^{t}a_{21}(s)\,ds}\leq Ke^{-\alpha(t-t_{0})} \quad \textnormal{for any $t\geq t_{0}$},
\end{equation}
and, by Proposition \ref{AVAD}, there exists
 a unique positive almost periodic solution $t\mapsto u^{*}(t)$ of (\ref{NHT3}). This implies that 
$t\mapsto p_{2,1}^{*}(t)=1/u^{*}(t)$ is an almost periodic solution of (\ref{2ep}) and thus we conclude that $t\mapsto (p_{1,1}^{*}(t),p_{2,1}^{*}(t))$ is an almost periodic solution of (\ref{tilman1-TV-n2}).

\medskip

\noindent\textit{Step 2: Uniform persistence of the species 2.} The second equation of (\ref{tilman1-TV-n2}) reads: 
\begin{displaymath}
    p_{2}'= p_{2}\left\{c_2(t)\left(h-p_{2}\right)-m_2(t)-[c_1(t)+c_{2}(t)]p_1(t) \right\}, 
\end{displaymath}
which, by using $u=1/p_{2}$, can be transformed into
\begin{equation}
\label{rough-0}
\dot{u}=b_{21}(t)u+c_{2}(t) \,\, \textnormal{where} \,\, b_{21}(t)=a_{21}(t)+[c_{1}(t)+c_{2}(t)]\{p_{1}(t)-p_{1,1}^{*}(t)\}.
\end{equation}

Now, let us define $\phi(t):=[c_{1}(t)+c_{2}(t)]\{p_{1}(t)-p_{1,1}^{*}(h,t)\}$ and note that the linear part of (\ref{rough-0}) is: 
\begin{displaymath}
\dot{z}=b_{21}(t)z  \quad \textnormal{or} \quad \dot{z}=[a_{21}(t)+\phi(t)]z,    
\end{displaymath}
where $\phi(t)=\textit{o}(1)$ since $p_{1}(t)-p_{1,1}^{*}(t)=\textit{o}(1)$ and
the colonization rates are bounded functions. This fact combined with $\mathcal{M}[a_{21}]<0$
and Lemma \ref{LT1} implies the existence $\theta>0$ such that any solution $t\mapsto u(t)$ of (\ref{rough-0}) verifies
$$
\limsup\limits_{t\to +\infty}u(t)\leq \theta\quad \textnormal{or equivalently} \quad
\liminf\limits_{t\to +\infty}p_{2}(t)\geq \theta^{-1},
$$
and the strong uniform persistence of $p_{2}$ follows.

\medskip

\noindent\textit{Step 3: End of proof.}
Let $u^{*}(t)=1/p_{2,1}^{*}(t)$ and $u(t)=1/p_{2}(t)$. We define the auxiliary function $v(t):=u(t)-u^{*}(t)$ and note that:
$$
v'=a_{21}(t)v+\phi(t)u(t),
$$
where the function $u(\cdot)\phi(\cdot)=\textit{o}(1)$ defined above also converges exponentially because 
$p_{1}(\cdot)-p_{1,1}^{*}(\cdot)=\textit{o}(1)$ converges exponentially and $t\mapsto u(t),c_{1}(t),c_{2}(t)$ are bounded.  
By using this fact combined with $\mathcal{M}[a_{21}]<0$ and Lemma \ref{LT2} we can deduce that $u(t)-u^{*}(t)=\textit{o}(1)$ and the convergence is exponential. Finally, note that
$$
|p_{2}(t)-p_{2,1}^{\ast}(t)|=\frac{|u(t)-u^{\ast}(t)|}{u(t)u^{\ast}(t)}\leq h^{2}|u(t)-u^{\ast}(t)|,
$$
which concludes the proof.
\end{proof}

The above result deserves some remarks: 

\medskip

Firstly, under the assumptions of Theorem \ref{GP}, a simple computation arising from (\ref{a2}) implies that the
averages of $t\mapsto p_{1,1}^{\ast}(t)$ and $t\mapsto p_{2,1}^{\ast}(t)$ verify the trade off condition: 
\begin{equation}
\label{tradeoff}
\mathcal{M}[(c_{1}+c_{2})p_{1,1}^{\ast}]+\mathcal{M}[c_{2}p_{2,1}^{\ast}]=\mathcal{M}[c_{2}](h-\mathcal{R}_{1}),
\end{equation}
then we can see that an increasing of the fraction of patches occupied by species 1 implies a declining of those occupied by the species 2 and vice versa. We point out that this fact is consistent with the competitive structure of the system (\ref{tilman1-TV-n2}).

Secondly, we know that the natural habitability threshold condition $\mathcal{R}_{1}<h$ 
ensures the strong persistence of the second species in absence of competition. Nevertheless, this condition is not enough to guarantee the uniform strong persistence under the presence of a superior competitor in some cases. Indeed, Theorem \ref{GP} states that if the patches occupied by the superior competitor converges to a positive almost periodic fraction, then a necessary condition to ensure the strong persistence of the species 2 is given by $\mathcal{R}_{1,1}^{\ast}=\mathcal{R}_{1}+\ell_{1,1}<h$. In this context, and emulating the literature of the autonomous framework \cite{Calcagno,Kinzig,Miller}, the interval $(\mathcal{R}_{1},\mathcal{R}_{1}+\ell_{1,1})$ will be called as the \textit{exclusion interval} or \textit{niche shadow}. The length of this is interval, namely $\ell_{1,1}$, is a constant
dependent of $h$ and $p_{1,1}^{\ast}(\cdot)$ and allow to see the threshold condition (\ref{tradeoff}) from another perspective.

Finally, the niche shadow condition $h\in (\mathcal{R}_{1},\mathcal{R}_{1}+\ell_{1,1})$ must imply the extinction
of the second species, which is discussed in the following result:

\begin{theorem}[Best competitor persists]
\label{BCP}
    Assume that $\mathcal{M}[c_{1}]<\mathcal{M}[c_{2}]$ is verified. If the fraction of habitable patches together with the colonization and extinction rates are such that the inequalities
    \begin{displaymath}
     \mathcal{R}_{0}<h<\mathcal{R}_{1,1}^{*},   
    \end{displaymath}
are also verified, then there exists a unique nontrivial almost periodic solution $t\mapsto (p_{1,1}^*(t),0)$, which is exponentially attractive for every positive solution $t\mapsto (p_1(t),p_2(t))$ of \eqref{tilman1-TV-n2}--\eqref{TH}, namely  
    \begin{displaymath}
\lim\limits_{t\to +\infty}|p_{1}(t)-p_{1,1}^{*}(t)| + |p_{2}(t)|=0.
\end{displaymath}
\end{theorem}

\begin{proof}
As $\mathcal{R}_{0}<h$, we know by Theorem \ref{Prop-1} that $t\mapsto p_{1,1}^{*}(t)$ is the unique non-trivial almost periodic solution of the first equation of (\ref{tilman1-TV-n2}) which attracts any solution
$p_{1}(t)$ at exponential rate. Now, the proof follows by proving that any component $p_{2}(t)$ converges exponentially to
zero.

We know that the second equation of (\ref{tilman1-TV-n2}) can be transformed 
via $v=1/p_{2}$ in the inhomogeneous equation
\begin{equation}
\label{a2-bis}
v'=[a_{21}(t)+\phi(t)]v+c_{2}(t), 
\end{equation}
where $t\mapsto a_{21}(t)$ is defined by (\ref{a2}) and
$\phi(t)=[c_1(t)+c_{2}(t)]\{p_1(t)-p_{1,1}^{*}(t)\}$, which is convergent to zero since $p_{1}(t)-p_{1,1}^{*}(t)=\textit{o}(1)$
and the colonization rates are bounded.

Moreover, as $h<\mathcal{R}_{1,1}^{*}$, by (\ref{resumen4}) we have that $\mathcal{M}[a_{21}]>0$. In 
addition, as $a_{21}(\cdot)$ is almost periodic and $\phi(t)=\textit{o}(1)$, the Lemma \ref{LT11} 
implies that $v(t)\geq K^{-1}e^{\alpha_{0}(t-t_{0})}$,  
now, as $p_{2}(t)=1/v(t)$. Finally we can deduce that
\begin{displaymath}
0\leq p_{2}(t) \leq  Ke^{-\alpha_{0}(t-t_{0})}   
\end{displaymath}
and the extinction of the metapopulation at exponential rate follows.
\end{proof}

\subsection{Strong persistence/extinction scenarios when $h< \mathcal{R}_{0}$}
Now, we will find sufficient conditions for the scenarios described in the second column of the Table \ref{Table1}. 
To this end, let us consider the almost periodic equation
\begin{equation}
\label{a3}
\dot{z}=a_{22}(t)z \quad \textnormal{with} \quad a_{22}(t)= -[h c_2(t)-m_2(t)].
\end{equation}

It is directly  verified that:
\begin{equation}
\label{resumen5}
\mathcal{M}[a_{22}]<0 \iff h\mathcal{M}[c_{2}]-\mathcal{M}[m_{2}]>0 \iff h>\mathcal{R}_{1},
\end{equation}
and
\begin{equation}
\label{resumen6}
\mathcal{M}[a_{22}]>0 \iff h\mathcal{M}[c_{2}]-\mathcal{M}[m_{2}]<0 \iff h<\mathcal{R}_{1},
\end{equation}
where $\mathcal{R}_{1}$ is the habitability threshold defined by (\ref{GTH}).

The next result says that the natural habitability threshold condition $\mathcal{R}_{1}<h$ 
assures the strong persistence of the second species even under the presence of a superior competitor
driven to the strong extinction. This follows from the absence of the niche shadow.

\begin{theorem}[Fugitive species persists]
\label{FSP}
Assume that $\mathcal{M}[c_{1}]<\mathcal{M}[c_{2}]$.
 If the fraction of habitable patches together with the colonization and extinction rates are such that the inequalities
    \begin{displaymath}
       \mathcal{R}_{1}<h < \mathcal{R}_{0},
    \end{displaymath}
are verified, then any solution $t\mapsto (p_{1}(t),p_{2}(t))$ of the system \eqref{tilman1-TV-n2}--\eqref{TH} verifies:
    \begin{displaymath}
\lim\limits_{t\to +\infty}|p_{1}(t)| + |p_{2}(t)-p_{2,3}^{*}(t)|=0,
\end{displaymath}
where the convergence is exponential and the almost periodic function $t\mapsto p_{2,3}^{*}(t)$ is defined by
\begin{equation}
\label{q2}
p_{2,3}^{*}(t)=\left(\int_{-\infty}^{t}e^{\int_{s}^{t}a_{22}(\tau)\,d\tau}c_{2}(s)\,ds\right)^{-1}.
\end{equation}
\end{theorem}

 \begin{proof}
The proof will be divided into several steps.

\noindent \textit{Step 1: Existence of an almost periodic solution.} As stated in Theorem \ref{Prop-1}, the inequality
$h < \mathcal{R}_{0}$ implies that $t\mapsto p_{1}(t)=\textit{o}(1)$ and this decay is exponential. Now, when replacing 
$t\mapsto p_{1}(t)$ by zero in the second equation, we obtain:
\begin{equation}
\label{2ep-AE}
p_{2}'= \displaystyle p_{2}\left\{c_2(t)\left(h-p_2\right)-m_2(t) \right\}.
\end{equation}

By proceeding analogously as in the proof of Theorem \ref{Prop-1} but considering the habitability threshold $\mathcal{R}_{1}
$, we can prove the existence of an almost periodic solution $t\mapsto p_{2,3}^{*}(t)$ of (\ref{2ep-AE}) which is exponentially attractive since $h>\mathcal{R}_{1}$ is equivalent to $\mathcal{M}[a_{3}]<0$.

\medskip  

\noindent\textit{Step 2: Strong uniform persistence of species 2.} The second equation of (\ref{tilman1-TV-n2}) is: 
\begin{displaymath}
    p_{2}'= p_{2}\left\{c_2(t)\left(h-p_{2}\right)-m_2(t)-[c_1(t)+c_{2}(t)]p_1(t) \right\}, 
\end{displaymath}
which, by using $u=1/p_{2}$, can be transformed into
\begin{equation}
\label{rough}
\dot{u}=b_{22}(t)u+c_{2}(t) \,\, \textnormal{where} \,\, b_{22}(t)=a_{22}(t)+[c_{1}(t)+c_{2}(t)]p_{1}(t).
\end{equation}

As in the proof of the previous Theorem, let us define $\psi(t)=[c_{1}(t)+c_{2}(t)]p_{1}(t)$ and notice that the linear part of (\ref{rough}) is 
\begin{displaymath}
\quad \dot{z}=[a_{22}(t)+\psi(t)]z.     
\end{displaymath}

As $p_{1}(t)=\textit{o}(1)$ and the colonization rates are bounded functions, it follows
that $\psi(t)=\textit{o}(1)$. By using this fact combined with $\mathcal{M}[a_{22}]<0$, the 
Lemma \ref{LT1} implies the existence of $\theta>0$ such that
$$
\limsup\limits_{t\to +\infty}u(t)\leq \theta
\quad \textnormal{or equivalently} \quad
\liminf\limits_{t\to +\infty}p_{2}(t)\geq \theta^{-1},
$$
and the strong uniform persistence of $p_{2}$ follows.

\medskip

\noindent\textit{Step 3: End of proof.}
Let $u^{*}(t)=1/p_{2,3}^{*}(t)$ and define $v(t)=u(t)-u^{*}(t)$. Note that
$v(\cdot)$ is a solution of the differential equation:
$$
v'=[a_{22}(t)+\psi(t)]v \quad \textnormal{with} \quad \psi(t)=[c_{1}(t)+c_{2}(t)]p_{1}(t).
$$

As $\mathcal{M}[a_{22}]<0$, $u(\cdot)$ is bounded and $\psi(t)=\textit{o}(1)$ exponentially, the Lemma \ref{LT2} implies that $v(t)=\textit{o}(1)$ and the convergence is exponential. Moreover, note that
$$
|p_{2,3}^{*}(t)-p_{2}(t)|=\frac{|u(t)-u^{*}(t)|}{u(t)u^{*}(t)}\leq h^{2}|v(t)|=\textit{o}(1),  
$$
and the exponential convergence follows. 
\end{proof}

\begin{theorem}[global strong extinction]
\label{GE}
    Assume that $\mathcal{M}[c_{1}]<\mathcal{M}[c_{2}]$ is verified. 
    If the fraction of habitable patches together with the colonization and extinction rates are such that the inequality
    \begin{displaymath}
    h < \min\left\{\mathcal{R}_{0},\mathcal{R}_{1}\right\}
    \end{displaymath}
  is verified, then every solution $t\mapsto (p_1(t),p_2(t))$ of \eqref{tilman1-TV-n2}--\eqref{TH} has the exponential convergence: 
    \begin{displaymath}
\lim\limits_{t\to +\infty}|p_{1}(t)| + |p_{2}(t)|=0.
\end{displaymath}
\end{theorem}

\begin{proof}
By Theorem \ref{Prop-1}, we know that $p_{1}(t)=\textit{o}(1)$. In addition, notice that the second equation of (\ref{tilman1-TV-n2}) has the form: 
\begin{displaymath}
    p_{2}'= p_{2}\left\{c_2(t)\left(h-p_{2}\right)-m_2(t)-[c_1(t)+c_{2}(t)]p_1(t) \right\}, 
\end{displaymath}
which, by using $u=1/p_{2}$, can be transformed into (\ref{rough}), namely,
\begin{equation}
\label{NHET}
\dot{u}=b_{22}(t)u+c_{2}(t) \,\, \textnormal{where} \,\, b_{22}(t)=a_{22}(t)+[c_{1}(t)+c_{2}(t)]p_{1}(t).
\end{equation}

As in the proof of Theorem \ref{FSP}, let us define the map $\psi(t)=[c_{1}(t)+c_{2}(t)]p_{1}(t)$ and notice that the linear part of the above equation is 
\begin{equation}
\label{LET}
\dot{z}=[a_{22}(t)+\phi(t)]z.     
\end{equation}

Let us recall that $h<\mathcal{R}_{1}$ is equivalent to $\mathcal{M}[a_{22}]>0$. Then the Proposition \ref{AVED} implies the existence of $\alpha>0$ and $C\geq 1$ such that
\begin{displaymath}
e^{\int_{t_{0}}^{t}a_{22}(s)\,ds}\geq C^{-1}e^{\alpha(t-t_{0})} \quad \textnormal{for any $t\geq t_{0}$}.
\end{displaymath}

On the other hand, as $\psi(t)=\textit{o}(1)$, the Proposition \ref{rough-a} implies the existence of 
$K>1$ and $\alpha_{0}\in (0,\alpha)$ such that
\begin{displaymath}
e^{\int_{t_{0}}^{t}a_{22}(s)\,ds}\geq K^{-1}e^{\alpha_{0}(t-t_{0})} \quad \textnormal{for any $t\geq t_{0}\geq 0$},
\end{displaymath}
which allows us to prove that 
$$
K^{-1}e^{\int_{t_{0}}^{t}[a_{22}(s)+\phi(s)]\,ds}\geq e^{\alpha_{0}(t-t_{0})}e^{\int_{t_{0}}^{t}\phi(s)\,ds} \quad \textnormal{for any $t\geq t_{0}$},
$$
then, for any $t>t_{0}\geq 0$, we can see that any solution 
of (\ref{NHET}) verifies:
\begin{displaymath}
\begin{array}{rcl}
u(t) &\geq & K^{-1}e^{\alpha_{0} (t-t_{0})}u(t_{0})+K^{-1}\int_{t_{0}}^{t}e^{\int_{t_{0}}^{t}[a_{3}(s)+\phi(s)]c_{2}(s)}\,ds \\\\
&\geq & e^{\alpha_{0} (t-t_{0})},
\end{array}
\end{displaymath}
since $c_{2}(\cdot)>0$ and $u(t_{0})\geq 1$. Then the result follows since we can deduce that 
$$
p_{2}(t)\leq e^{-\alpha_{0}(t-t_{0})} \quad \textnormal{for any $t>t_{0}\geq 0$}.
$$

\end{proof}

\section{The case with $n\geq 3$}

The proofs carried out in the previous sections with $n=1,2$ show 
a pattern of regularity, which we will try to formalize in this section by 
describing a general procedure for cases beyond $n\geq 3$.

\medskip
\subsection{The general case} We will describe a formal procedure in several steps:

\medskip

\noindent\textsf{Step 1:}
If $n=k-1$, let us assume that the value of $h\in (0,1)$ implies that for any solution $t\mapsto p(t)=(p_{1}(t),\ldots,p_{k-1}(t))$ of the system 
\begin{equation}
\label{tilman1-TV-k}
p_{i}'=c_i(t) p_i\left(h-\sum_{j=1}^{i}p_j \right)-m_i(t)p_i-\sum_{j=1}^{i-1}c_j(t) p_i p_j, \quad
\textnormal{with $i=1,\ldots,k-1$},
\end{equation}
with $p_{1}(0)+\ldots+p_{k-1}(0)\in (0,h)$, there exists a unique $\ell \in \{1,\ldots,2^{k-1}\}$
and $k-1$ nonnegative almost periodic functions $t\mapsto p_{j,\ell}^{\ast}(t)$ with $j \in \{1,\ldots,k-1\}$ such that
\begin{displaymath}
\lim\limits_{t\to +\infty}|p_{1}(t)-p_{1,\ell}^{\ast}(t)|+\ldots+|p_{k-1}(t)-p_{k-1,\ell}^{\ast}(t)|=0,  
\end{displaymath}
and the convergence is exponential.

In other words, we are assuming the existence of $2^{k-1}$ scenarios of strong persistence/extinction described by the almost periodic functions $t\mapsto p_{\ell}^{\ast}(t)$:
$$
p_{\ell}^{*}(t)=(p_{1,\ell}^{*}(t),p_{2,\ell}^{*}(t),\ldots,p_{k-1,\ell}^{*}(t))
\quad \textnormal{where $\ell \in \{1,\ldots,2^{k-1}\}$},
$$
and any component $t\mapsto p_{j,\ell}^{\ast}(t)$ is either
zero or a positive almost periodic function.

\medskip
\noindent\textsf{Step 2:}
If $n=k$, the previous step ensures that the behavior of the first $k-1$ components of the system
\begin{equation}
\label{n=k}
p_{i}'=c_i(t) p_i\left(h-\sum_{j=1}^{i}p_j \right)-m_i(t)p_i-\sum_{j=1}^{i-1}c_j(t) p_i p_j, \quad
\textnormal{with $i=1,\ldots,k$},
\end{equation}
is well known, while the dynamics of the last equation
\begin{equation}
\label{mishvsof}
p_{k}'=c_{k}(t) p_k\left(h-p_{k}-\sum_{j=1}^{k-1}p_j \right)-m_{k}(t)p_{k}-p_{k}\sum_{j=1}^{k-1}c_j(t)p_j,
\end{equation}
remains to be studied. In order to do that, we know by the previous step that for any solution
$p(t)=(p_{1}(t),\ldots,p_{k-1}(t),p_{k}(t))$ of (\ref{n=k}), there exists a unique
$\ell \in \{1,\ldots,2^{k-1}\}$ such that
\begin{displaymath}
p_{j}(t)-p_{j,\ell}^{\ast}(t)=\textit{o}(1) \quad \textnormal{where $j\in \{1,\ldots,k-1\}$},
\end{displaymath}
and the convergence is exponential. This allows us to rewrite the equation (\ref{mishvsof}) as follows:
\begin{equation}
\label{mishvsof++}
\begin{array}{rcl}
p_{k}' & = & \displaystyle c_{k}(t) p_k\left(h-p_{k}-\sum_{j=1}^{k-1}[p_{j,\ell}^{\ast}(t)+\textit{o}(1)] \right)-m_{k}(t)p_{k}\\
&& \displaystyle -p_{k}\sum_{j=1}^{k-1}c_j(t)[p_{j,\ell}^{\ast}(t)+\textit{o}(1)],
\end{array}
\end{equation}
and we conclude that (\ref{mishvsof}) is asymptotically 
equivalent to the almost periodic scalar equation:
\begin{equation}
\label{mishvsof+}
p_{k}'=c_{k}(t) p_k\left(h-p_{k}-\sum_{j=1}^{k-1}p_{j,\ell}^{\ast}(t) \right)-m_{k}(t)p_{k}-p_{k}\sum_{j=1}^{k-1}c_j(t)p_{j,\ell}^{\ast}(t).
\end{equation}

\medskip
\noindent\textsf{Step 3:} We know that the change of variables
$u_{k}=1/p_{k}$ transforms (\ref{mishvsof+}) into the almost periodic inhomogeneous equation:
\begin{equation}
\label{linmish}
u_{k}'=\underbrace{-\left[hc_{k}(t)-m_{k}(t)-\sum_{j=1}^{k-1}\{c_{k}(t)+c_j(t)\}p_{j,\ell}^{\ast}(t)\right]}_{:=a_{k\ell}(t)}u_{k}+c_{k}(t).
\end{equation}

By Proposition \ref{AVED} we know that the linear part of the above equation has an exponential dichotomy on $\mathbb{R}$ if and only if
\begin{displaymath}
\mathcal{M}[a_{k\ell}]<0 \iff h>\mathcal{R}_{k-1,\ell}^{\ast} \quad \textnormal{or} \quad
\mathcal{M}[a_{k\ell}]>0 \iff h<\mathcal{R}_{k-1,\ell}^{\ast},
\end{displaymath}
where the habitability threshold $\mathcal{R}_{k,\ell}^{\ast}$ is defined by
\begin{displaymath}
\mathcal{R}_{k-1,\ell}^{\ast}:=\mathcal{R}_{k-1}+\ell_{k-1,\ell}\quad \textnormal{with} \quad \ell_{k-1,\ell}=\frac{\mathcal{M}\left[\sum_{j=1}^{k-1}\{c_{k}(t)+c_j(t)\}p_{j,\ell}^{\ast}(t)\right]}{\mathcal{M}[c_{k}]}.    
\end{displaymath}

By Proposition \ref{AVAD} and Corollary \ref{cor1} we know that if $\mathcal{M}[a_{k\ell}]<0$ then (\ref{linmish}) has a unique almost periodic solution described by
\begin{equation}
\label{pit-sc}
u_{k,\ell}^{\ast}(t)=\int_{-\infty}^{t}e^{-\int_{s}^{t}a_{k\ell}(\tau)\,d\tau}c_{k}(s)\,ds,
\end{equation}
while any other solution $t\mapsto u_{k}(t)$ verifies $|u_{k}-u_{k}^{\ast}|=\textit{o}(1)$ and the convergence
is exponential.

\medskip
\noindent\textsf{Step 4:}  Similarly, we know that the change of variables
$v_{k}=1/p_{k}$ transforms (\ref{mishvsof++}) into the inhomogeneous asymptotically almost periodic equation:
\begin{equation}
\label{FEQ}
v_{k}'=-[a_{k\ell}(t)+\psi(t)]v_{k}+c_{k}(t),  
\end{equation}
where the function $\psi(t)=\textit{o}(1)$ converges exponentially.

\medskip

 \textsf{Case A:}  $\mathcal{M}[a_{k\ell}]<0$, the  Lemma \ref{LT1} implies the existence $\theta>0$ such that any solution $t\mapsto v_{k}(t)$ of (\ref{FEQ}) verifies
$$
\limsup\limits_{t\to +\infty}v_{k}(t)\leq \theta\quad \textnormal{or equivalently} \quad
\liminf\limits_{t\to +\infty}p_{k}(t)\geq \theta^{-1},
$$
and the strong uniform persistence of $p_{k}$ follows. Now, as $u_{k,\ell}^{*}(t)=1/p_{k,\ell}^{*}(t)$ and $v(t)=1/p_{k}(t)$, we define the auxiliary function $w(t):=v_{k}(t)-u_{k,\ell}^{*}(t)$ and note that:
$$
w'=-a_{k\ell}(t)w+\psi(t)v_{k}(t).
$$

As $t\mapsto v_{k}(t)$ is bounded on $[0,+\infty)$ we have that $t\mapsto \psi(t)v_{k}(t)=\textit{o}(1)$ and the
$\psi(t)$ convergence is exponential. By using this fact combined with $\mathcal{M}[a_{k\ell}]>0$, the Lemma \ref{LT2} 
implies that $v_{k}(t)-u_{k,\ell}^{*}(t)=\textit{o}(1)$ and the convergence is exponential. Finally, it can be deduced that
$$
|p_{k}(t)-p_{k,\ell}^{\ast}(t)|=p_{k}(t)p_{k,\ell}^{*}(t)|v_{k}(t)-u_{k,\ell}^{*}(t)|<|v_{k}(t)-u_{k,\ell}^{*}(t)|=\textit{o}(1),
$$
then we have that any solution of (\ref{mishvsof}) converges exponentially to $p_{k,\ell}^{\ast}(\cdot)$ and the
strong persistence of the $k$--th species in ensured.

\medskip

\textsf{Case B:}  $\mathcal{M}[a_{k\ell}]>0$, the  Lemma \ref{LT11} implies the existence $K>0$ 
and $\alpha_{0}>0$
such that any solution $t\mapsto v_{k}(t)$ of (\ref{FEQ}) verifies $v_{k}(t)>K^{-1}e^{\alpha(t-t_{0})}$, which implies that 
any solution $t\mapsto p_{k}(t)$ of (\ref{mishvsof}) verifies $p_{k}(t)\leq Ke^{-\alpha_{0}(t-t_{0})}$ and the
strong persistence of the $k$--th species in ensured.

\medskip
\noindent \textsf{Final Step:} In the above steps, we have proved that for any fixed $\ell \in \{1,\ldots,2^{k-1}\}$. There
exists to possible scenarios of strong persistence/extinction of (\ref{mishvsof}) provided that $h>\mathcal{R}_{k,\ell}^{\ast}$
or $h<\mathcal{R}_{k,\ell}^{\ast}$. In consequence, there exist $2^{k}$ scenarios of strong persistence/extinction for the metapopulation (\ref{n=k}).

\medskip 

\subsection{The case $n=3$} We will provide a detailed application of the above procedure for the
case $n=3$, where the metapopulation model described by (\ref{tilman1-TV}) becomes: 
\begin{equation}
\label{tilman1-TV-n3}
\left\{
\begin{array}{rcl}
\displaystyle p_{1}' &=& p_{1}\left\{c_1(t)\left(h-p_1\right)-m_1(t)\right\} \\
\displaystyle p_{2}' &=& p_{2}\left\{c_2(t)\left(h-p_{1}-p_{2} \right)-m_2(t)-c_1(t) p_1\right\}  \\
\displaystyle p_{3}' &=& p_{3}\left\{c_{3}(t)\left(h-p_{1}-p_{2}-p_{3}\right)-m_{3}(t)
-c_{1}(t)p_{1}-c_{2}(t)p_{2}\right\},
\end{array}\right.
\end{equation}
with initial conditions 
\begin{equation}
\label{THn3}
0<p_{1}(t_{0})+p_{2}(t_{0})+p_{3}(t_{0})<h.
\end{equation}

\subsubsection{Scenarios when $\max\{\mathcal{R}_{0},\mathcal{R}_{1}^{*}\}<h$}

In this case, the behavior of the first two equations of (\ref{tilman1-TV-n3}) is described by Theorem \ref{GP} while the third equation is asymptotically equivalent to the almost periodic equation:
\begin{displaymath}
p_{3}'=p_{3}\left\{c_{3}(t)\left(h-p_{3}-p_{\tiny{1,1}}^{*}(t)-p_{2,1}^{*}(t)\right)-m_{3}(t)
-c_{1}(t)p_{1,1}^{*}(t)-c_{2}(t)p_{2,1}^{*}(t)\right\},
\end{displaymath}
where the positive almost periodic functions $t\mapsto p_{1,1}^{*}(t)$ and $t\mapsto p_{2,1}^{*}(t)$
are defined in (\ref{sol-periodique}) and (\ref{p2}) respectively. This equation can be transformed via $u=1/p_{3}$ into
\begin{equation}
\label{n3-1}
u'=a_{31}(t)u+c_{3}(t),
\end{equation}
where
\begin{displaymath}
 a_{31}(t)=-[c_{3}(t)\{h-p_{1,1}^{*}(t)-p_{2,1}^{*}(t)\}-m_{3}(t)-c_{1}(t)p_{1,1}^{*}(t)-c_{2}(t)p_{2,1}^{*}(t)].
\end{displaymath}

By following an analogous procedure to the previous section, we can see that the strong persistence/extinction of the third species is determined by either the positiveness or negativeness of the average $\mathcal{M}[a_{31}]$:
\begin{equation}
\label{sp3}
\mathcal{M}[a_{31}]<0  \iff  \mathcal{R}_{2,1}^{*}<h \quad \textnormal{and} \quad
\mathcal{M}[a_{31}]>0  \iff  \mathcal{R}_{2,1}^{*}>h,
\end{equation}
where 
\begin{displaymath}
\mathcal{R}_{2,1}^{*}:=\mathcal{R}_{2}+\ell_{2,1} \quad
\textnormal{with} \,\, \ell_{2,1}:=\frac{\mathcal{M}[p_{1,1}^{*}(c_{1}+c_{3})]+\mathcal{M}[p_{2,1}^{*}(c_{2}+c_{3})]}{\mathcal{M}[c_{3}]}
\end{displaymath}
and $\mathcal{R}_{2}$ is described by (\ref{GTH}) with $i=3$. Notice that,
this defines a extinction interval or niche shadow $[\mathcal{R}_{2},\mathcal{R}_{2}+\ell_{2,1}]$ for the species 3, which is induced
by the superior competitors.

Now, let $t\mapsto (p_{1}(t),p_{2}(t),p_{3}(t))$ be a solution of (\ref{tilman1-TV-n3})--(\ref{THn3}), then
we have the following scenarios where the convergence is exponential:
\begin{itemize}
\item[$\bullet$] If $\max\left\{\mathcal{R}_{0},\mathcal{R}_{1,1}^{*},\mathcal{R}_{2,1}^{*}\right\}<h$, then
\begin{displaymath}
\lim\limits_{t\to +\infty}|p_{1}(t)-p_{1,1}^{*}(t)|+|p_{2}(t)-p_{2,1}^{*}(t)|+|p_{3}(t)-p_{3,1}^{*}(t)|=0,    
\end{displaymath}
\item[$\bullet$] If $\max\left\{\mathcal{R}_{0},\mathcal{R}_{1,1}^{*}\right\}<h<\mathcal{R}_{2,1}^{*}$, then
\begin{displaymath}
\lim\limits_{t\to +\infty}|p_{1}(t)-p_{1,1}^{*}(t)|+|p_{2}(t)-p_{2,1}^{*}(t)|+|p_{3}(t)|=0  \quad \textnormal{(note that $p_{3,2}^{*}(t)=0$)},    
\end{displaymath}
\end{itemize}
and the almost periodic functions $t\mapsto p_{i,1}^{*}(h,t)$ are respectively described by (\ref{sol-periodique})
and (\ref{p2}) for $i=1$ and $i=2$, while
\begin{displaymath}
p_{3,1}^{*}(t)=\left( \int_{-\infty}^{t}e^{\int_{s}^{t}a_{31}(r)\,dr}c_{3}(s)\,ds\right)^{-1}.
\end{displaymath}

\subsubsection{Scenarios when $\mathcal{R}_{0}<h<\mathcal{R}_{1,1}^{*}$}
The behavior of the first two equations of (\ref{tilman1-TV-n3}) is described by Theorem \ref{BCP} while the third equation is asymptotically equivalent to the almost periodic equation:
\begin{displaymath}
p_{3}'=p_{3}\left\{c_{3}(t)\left(h-p_{3}-p_{1,1}^{*}(t)\right)-m_{3}(t)
-c_{1}(t)p_{1,1}^{*}(t)\right\},
\end{displaymath}
where the almost periodic function $t\mapsto p_{1,1}^{*}(t)$ 
is defined in (\ref{sol-periodique}). This equation is transformed into
\begin{equation}
\label{n3-2}
u'=a_{32}(t)u+c_{3}(t) \,\, \textnormal{where} \,\,\, a_{32}(t)=-[hc_{3}(t)-m_{3}(t)-\{c_{1}(t)+c_{3}(t)\}p_{1,1}^{*}(t)].
\end{equation}

As before, the strong persistence/extinction of the third species is determined by either the positiveness or negativeness of the average $\mathcal{M}[a_{32}]$:
\begin{equation}
\label{sp31}
\mathcal{M}[a_{32}]<0  \iff  \mathcal{R}_{2,2}^{*}<h \quad \textnormal{and} \quad
\mathcal{M}[a_{32}]>0  \iff  \mathcal{R}_{2,2}^{*}>h,
\end{equation}
where 
\begin{displaymath}
\mathcal{R}_{2,2}^{*}:=\mathcal{R}_{2}+\ell_{2,2} \quad \textnormal{where} \quad \ell_{2,2}=\frac{\mathcal{M}[p_{1}^{*}(c_{1}+c_{3})]}{\mathcal{M}[c_{3}]}
\end{displaymath}
and note that this defines a extinction interval or niche shadow $[\mathcal{R}_{2},\mathcal{R}_{2}+\ell_{2,2}]$ for the species 3, which is induced
by a superior competitor, namely, the species 1.

Let $t\mapsto (p_{1}(t),p_{2}(t),p_{3}(t))$ be a solution of (\ref{tilman1-TV-n3})--(\ref{THn3}) then
the following sce\-na\-rios with exponential convergence are verified:
\begin{itemize}
\item[$\bullet$] If $\max\{\mathcal{R}_{0},\mathcal{R}_{2,2}^{*}\}<h<\mathcal{R}_{1,1}^{*}$, then
\begin{displaymath}
\lim\limits_{t\to +\infty}|p_{1}(t)-p_{1,1}^{*}(t)|+|p_{2}(t)|+|p_{3}(t)-p_{3,3}^{*}(t)|=0.    
\end{displaymath}
\item[$\bullet$] If $\mathcal{R}_{0}<h<\min\{\mathcal{R}_{1,1}^{*},\mathcal{R}_{2,2}^{*}\}$, then
\begin{displaymath}
\lim\limits_{t\to +\infty}|p_{1}(t)-p_{1,1}^{*}(t)|+|p_{2}(t)|+|p_{3}(t)|=0   \quad \textnormal{(note that $p_{3,4}^{*}(t)=0$)},   
\end{displaymath}
\end{itemize}
and the almost periodic functions $t\mapsto p_{i,1}^{*}(t)$ are respectively described by (\ref{sol-periodique})
and (\ref{p2}) for $i=1$ and $i=2$, while the almost periodic function $t\mapsto p_{3,3}^{*}(t
)$ is defined by:
\begin{displaymath}
p_{3,3}^{*}(t)=\left( \int_{\infty}^{t}e^{\int_{s}^{t}a_{32}(r)\,dr}c_{3}(s)\,ds\right)^{-1}.  \end{displaymath}

\subsubsection{Scenarios when $\mathcal{R}_{1}<h<\mathcal{R}_{0}$}
 The behavior of the first two equations of (\ref{tilman1-TV-n3}) is described by Theorem \ref{FSP} while the third equation is asymptotically equivalent to the almost periodic equation:
\begin{displaymath}
p_{3}'=p_{3}\left\{c_{3}(t)\left(h-p_{3}-p_{2,3}^{*}(t)\right)-m_{3}(t)
-c_{2}(t)p_{2,3}^{*}(t)\right\},
\end{displaymath}
where the almost periodic function $t\mapsto p_{2,3}^{*}(t)$
is defined by (\ref{q2}). This equation can be transformed into
\begin{equation}
\label{n3-32}
u'=a_{33}(t)u+c_{3}(t) \quad \textnormal{where} \quad 
 a_{33}(t)=-[hc_{3}(t)-m_{3}(t)-\{c_{2}(t)+c_{3}(t)\}p_{2,3}^{*}(t)].
\end{equation}

The strong persistence/extinction of the third species is determined by either the positiveness or negativeness of the average $\mathcal{M}[a_{33}]$:
\begin{equation}
\label{sp30}
\mathcal{M}[a_{33}]<0  \iff  \mathcal{R}_{2,3}^{*}<h \quad \textnormal{and} \quad
\mathcal{M}[a_{33}]>0  \iff  \mathcal{R}_{2,3}^{*}>h,
\end{equation}
where 
\begin{displaymath}
\mathcal{R}_{2,3}^{*}:=\mathcal{R}_{2}+\ell_{2,3} \quad \textnormal{with} \quad \ell_{2,3}=\frac{\mathcal{M}[p_{2,3}^{*}(c_{2}+c_{3})]}{\mathcal{M}[c_{3}]},
\end{displaymath}
which defines an extinction interval or niche shadow $[\mathcal{R}_{2},\mathcal{R}_{2}+\ell_{2,3}]$ for the species 3, which is induced
by a superior competitor, namely, the species 2.

Let $t\mapsto (p_{1}(t),p_{2}(t),p_{3}(t))$ be a solution of (\ref{tilman1-TV-n3})--(\ref{THn3}), then we have the following
scenarios with exponential convergence:
\begin{itemize}
\item[$\bullet$] If $\max\{\mathcal{R}_{1},\mathcal{R}_{2,3}^{*}\}<h<\mathcal{R}_{0}$, then
\begin{displaymath}
\lim\limits_{t\to +\infty}|p_{1}(t)|+|p_{2}(t)-p_{2,3}^{*}(t)|+|p_{3}(t)-p_{3,5}^{*}(t)|=0.    
\end{displaymath}
\item[$\bullet$] If $\mathcal{R}_{1}<h<\min\{\mathcal{R}_{0},\mathcal{R}_{2,3}^{*}\}$, then
\begin{displaymath}
\lim\limits_{t\to +\infty}|p_{1}(t)|+|p_{2}(t)-p_{2,3}^{*}(t)|+|p_{3}(t)|=0    \quad \textnormal{(note that $p_{3,6}^{*}(t)=0$)}, 
\end{displaymath}
\end{itemize}
where the almost periodic functions $t\mapsto p_{2,3}^{*}(t)$ and $t\mapsto p_{3,5}^{*}(t)$ are described by
(\ref{q2}) and:
\begin{displaymath}
p_{3,5}^{*}(t)=\left( \int_{-\infty}^{t}e^{\int_{s}^{t}a_{33}(r)\,dr}c_{3}(s)\,ds\right)^{-1}.  \end{displaymath}

\subsubsection{Scenarios when $h<\min\{\mathcal{R}_{0},\mathcal{R}_{1}\}$}
 The behavior of the first two equations of (\ref{tilman1-TV-n3}) is described by Theorem \ref{GE} while the third equation is asymptotically equivalent to the almost periodic equation:
\begin{displaymath}
p_{3}'=p_{3}\left\{c_{3}(t)\left(h-p_{3}\right)-m_{3}(t)\right\},
\end{displaymath}
which can be transformed into
\begin{equation}
\label{n3-33}
u'=a_{34}(t)u+c_{3}(t) \quad \textnormal{where} \quad
 a_{34}(t)=-[hc_{3}(t)-m_{3}(t)].
\end{equation}

The strong persistence/extinction of the third species is determined by either the positiveness or negativeness of the average $\mathcal{M}[a_{34}]$:
\begin{equation}
\label{sp4}
\mathcal{M}[a_{34}]<0  \iff  \mathcal{R}_{2}<h \quad \textnormal{and} \quad
\mathcal{M}[a_{34}]>0  \iff  \mathcal{R}_{2}>h.
\end{equation}

Let $t\mapsto (p_{1}(t),p_{2}(t),p_{3}(t))$ be a solution of (\ref{tilman1-TV-n3})--(\ref{THn3}) then we have the
following scenarios with exponential convergence:
\begin{itemize}
\item[$\bullet$] If $\mathcal{R}_{2}<h<\min\{\mathcal{R}_{0},\mathcal{R}_{1}\}$, then
\begin{displaymath}
\lim\limits_{t\to +\infty}|p_{1}(t)|+|p_{2}(t)|+|p_{3}(t)-p_{3,7}^{*}(t)|=0.    
\end{displaymath}
\item[$\bullet$] If $h<\min\{\mathcal{R}_{0},\mathcal{R}_{1},\mathcal{R}_{2}\}$, then
\begin{displaymath}
\lim\limits_{t\to +\infty}|p_{1}(t)|+|p_{2}(t)|+|p_{3}(t)|=0  \quad \textnormal{(note that $p_{3,8}^{*}(t)=0$)}
\end{displaymath}
\end{itemize}
where the almost periodic function $t\mapsto p_{3,7}^{*}(t)$ is defined by:
\begin{displaymath}
p_{3,7}^{*}(t)=\left( \int_{\infty}^{t}e^{\int_{s}^{t}a_{34}(r)\,dr}c_{3}(s)\,ds\right)^{-1}.  \end{displaymath}

The $2^{3}$ Strong persistence/Strong extinction scenarios are summarized in the Figure \ref{sce3}

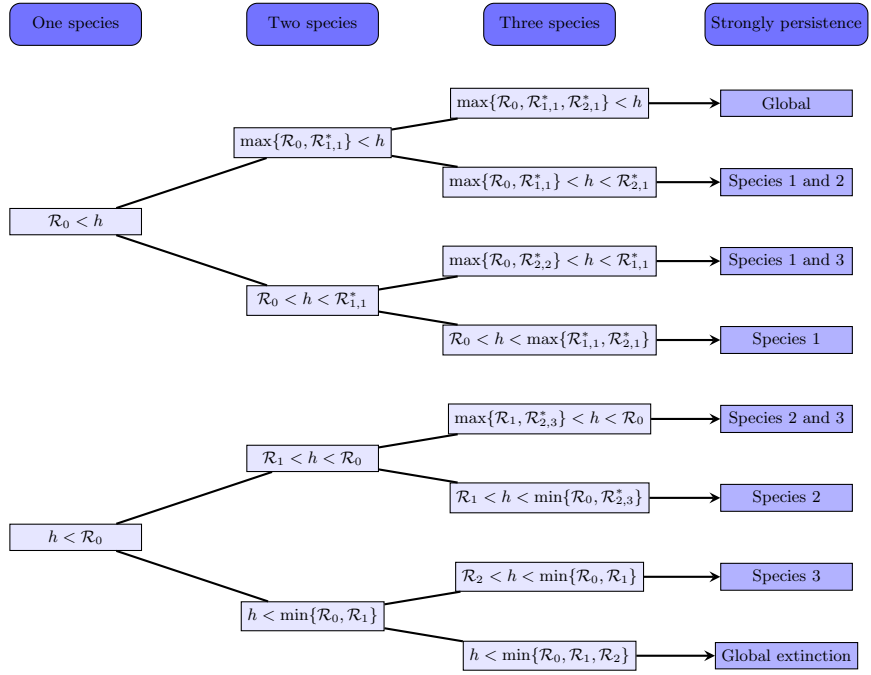
\begin{figure}[h]
    \centering
    \begin{tikzpicture}[node distance=1.5cm, scale=0.7, transform shape]

\node (species2) [startstop] {Two species};
\node (species3) [startstop, right of = species2, xshift=3 cm] {Three species};
\node (species1) [startstop, left of = species2, xshift=-3cm] {One species};
\node (persistence) [startstop, right of = species3, xshift=3cm] {Strongly persistence};
\node (global) [process, below of = persistence] {Global};
\node (12) [process, below of = global] {Species 1 and 2};
\node (13) [process, below of = 12] {Species 1 and 3};
\node (1) [process, below of = 13] {Species 1};

\node (23) [process, below of = 1] {Species 2 and 3};
\node (2) [process, below of = 23] {Species 2};
\node (3) [process, below of = 2] {Species  3};
\node (0) [process, below of = 3] {Global extinction};

\node (cg) [decision, left of = global, xshift=-3cm] {$\max\{ \mathcal{R}_0,\mathcal{R}_{1,1}^*,\mathcal{R}_{2,1}^*\}<h$};
\node (c12) [decision, left of = 12, xshift=-3cm] {$\max\{ \mathcal{R}_0,\mathcal{R}_{1,1}^*\}<h<\mathcal{R}_{2,1}^*$};
\node (c13) [decision, left of = 13, xshift=-3cm] {$\max\{ \mathcal{R}_0,\mathcal{R}_{2,2}^{*}\}<h<\mathcal{R}_{1,1}^*$};
\node (c1) [decision, left of = 1, xshift=-3cm] {$ \mathcal{R}_0<h<\max\{\mathcal{R}_{1,1}^*,\mathcal{R}_{2,1}^{*}\}$};

\node (c23) [decision, left of = 23, xshift=-3cm] {$\max\{ \mathcal{R}_1,\mathcal{R}_{2,3}^{*}\}<h<\mathcal{R}_0$};
\node (c2)[decision, left of = 2, xshift=-3cm] {$ \mathcal{R}_1<h<\min\{\mathcal{R}_0,\mathcal{R}_{2,3}^{*}\}$};
\node (c3)[decision, left of = 3, xshift=-3cm] {$ \mathcal{R}_2<h<\min\{\mathcal{R}_0,\mathcal{R}_1\}$};
\node (c0) [decision, left of = 0, xshift=-3cm] {$ h<\min\{\mathcal{R}_0,\mathcal{R}_1,\mathcal{R}_2\}$};

\node (cg2) [decision, left of = cg, xshift=-3cm, yshift=-0.75cm] {$\max\{ \mathcal{R}_0,\mathcal{R}_{1,1}^*\}<h$};
\node (c132) [decision, left of = c13, xshift=-3cm, yshift=-0.75cm] {$ \mathcal{R}_0<h<\mathcal{R}_{1,1}^*$};

\node (c232) [decision, left of = c23, xshift=-3cm, yshift=-0.75cm] {$\mathcal{R}_1<h<\mathcal{R}_0$};
\node (c32)[decision, left of = c3, xshift=-3cm, yshift=-0.75cm] {$ h<\min\{\mathcal{R}_0,\mathcal{R}_1\}$};
\node (cg21) [decision, left of = cg2, xshift=-3cm, yshift=-1.5cm] {$ \mathcal{R}_0<h$};

\node (c2321) [decision, left of = c232, xshift=-3cm, yshift=-1.5cm] {$h<\mathcal{R}_0$};

\draw [arrow] (cg21) -- (cg2) -- (cg)--(global);
\draw [arrow] (cg2) -- (c12)--(12);

\draw [arrow] (c132) -- (c13)--(13);
\draw [arrow] (cg21) --(c132) -- (c1)--(1);

\draw [arrow] (c2321) --(c232) -- (c23)--(23);
\draw [arrow] (c232) -- (c2)--(2);

\draw [arrow] (c32) -- (c3)--(3);
\draw [arrow] (c2321) --(c32) -- (c0)--(0);

\end{tikzpicture}
    \caption{Summary of conditions for the different scenarios of strong persistence in the case $n=3$.}
    \label{sce3}
\end{figure}

\section{Some Scenarios with weak extinction at some level}
As we have seen, the strong persistence/extinction results of the previous sections for $n=1,2,3$ were based
on the study of non homogeneous linear equations where the linear
part is an almost periodic function having nonzero average. In this section, we will drop
this assumption for some averages,
so the corresponding linear equation has not an exponential dichotomy
and, consequently, the results from section 2 cannot be used to study the metapopulation model at some step. 

\subsection{Technical results}

The following results are related to a linear scalar equation
and its nonhomogenous perturbation in absence of exponential dichotomy.

\begin{lemma}
\label{L3}
Let us consider an almost periodic function $\mathfrak{a}\colon \mathbb{R}\to \mathbb{R}$ and a continuous function
 $\psi\colon [0,+\infty)\to \mathbb{R}$. If $\mathcal{M}[\mathfrak{a}]=0$ and 
$\psi(t)=\textit{o}(1)$, then the perturbed equation
\begin{equation}
\label{perturbada}
\dot{x}=[\mathfrak{a}(t)+\psi(t)]x    
\end{equation}
does not has an exponential dichotomy on $[0,+\infty)$.
\end{lemma}

\begin{proof}
In order to conclude the result, we will prove a stronger fact, namely, $\Sigma_{[0,+\infty)}(\mathfrak{a}+\psi)=\{0\}$
which implies, according to Definition \ref{DEFSPEC}, that (\ref{perturbada}) has not an exponential dichotomy
on $[0,+\infty)$.

By Proposition \ref{SPEC}, we know that $\Sigma_{[0,+\infty)}(\mathfrak{a}+\psi)=[\beta^{-}(\mathfrak{a}+\psi),\beta^{+}(\mathfrak{a}+\psi)]$. Now, by using the
properties of lower and upper limits, we can deduce that:
$$
\beta^-(\mathfrak{a}) + \beta^-(\psi) \le \beta^-(\mathfrak{a}+\psi) \le \beta^+(\mathfrak{a}+\psi)
\le \beta^+(\mathfrak{a}) + \beta^+(\psi).
$$

Moreover, as $\mathfrak{a}(\cdot)$ is almost periodic, we know that the Bohl exponents coincide with the average $\mathcal{M}[\mathfrak{a}]$ which implies  that $\beta^+(\mathfrak{a})=\beta^-(\mathfrak{a})=0$ and we only need to verify that $\beta^+(\psi)=\beta^-(\psi)=0$. Indeed, for $\varepsilon >0$ fix $s_0$ such that $|\psi(s)| <\varepsilon$ for $s>s_0$, then for arbitrary $t>s>s_0$ we obtain
$$\left|\frac 1{t-s}\int_s^t\psi(r)\, dr\right| \le \varepsilon 
$$
and the result holds trivially.
\end{proof}

\begin{lemma}
\label{L33}
Let $\mathfrak{a}\colon \mathbb{R}\to \mathbb{R}$ be almost periodic with $\mathcal{M}[\mathfrak{a}]=0$ and consider two bounded continuous functions: a positive almost periodic function $\mathfrak{b}\colon \mathbb{R}\to (0,+\infty)$ 
{with {
    $\displaystyle\liminf_{t\to +\infty}\mathfrak b(t)=\mathfrak b_{0}>0$}}
and $\psi\colon [0,+\infty)\to \mathbb{R}$ where 
$\psi(t)=\textit{o}(1)$, then any solution of the perturbed equation
\begin{equation}
\label{perturbada+}
\dot{u}=[\mathfrak{a}(t)+\psi(t)]u+\mathfrak{b}(t)    
\end{equation}
with positive initial condition at $t=0$ is upper unbounded on $[0,+\infty)$.
\end{lemma}

\begin{proof} The proof will be made by contradiction. Indeed, if we assume the existence of an upper bounded solution $t\mapsto \hat{u}(t)$ of (\ref{perturbada+}) with positive
initial condition at $t=0$, it follows that
$0<\hat{u}(t)\leq M_{0}$ for any $t\geq 0$ because any solution of (\ref{perturbada+}) with positive initial condition at $t=0$ remains positive for any $t\geq 0$.

By using the boundedness of $\hat{u}(\cdot)$ combined with
$\liminf\limits_{t\to +\infty}\mathfrak{b}(t)=\mathfrak{b}_{0}>0$, the statement 4 from Proposition \ref{Prop-Cop} implies that the linear equation $\dot{u}=[\mathfrak{a}(t)+\psi(t)]u$ has an exponential dichotomy on $[0,+\infty)$, obtaining a contradiction with Lemma \ref{L3}.
\end{proof}

{The above result only proves that the solutions $t\mapsto u(t)$ of (\ref{perturbada+}) verify 
$$
\limsup\limits_{t\to +\infty}u(t)=+\infty,
$$ 
which arises the following question: Can we ensure that $\liminf\limits_{t\to +\infty}u(t)=+\infty$?. The next result provides a partial answer}

{\begin{lemma}
\label{Lper}
Under the assumptions of Lemma \ref{L33} restricted to $T$--periodic functions $\mathfrak{a}(\cdot)$ and $\mathfrak{b}(\cdot)$ if follows that $\lim\inf\limits_{t\to +\infty}u(t)=+\infty$.
\end{lemma}}

\begin{proof}
The proof will be made by contradiction by assuming the existence of $M>0$ such that $\liminf\limits_{t\to+\infty}u(t)=M>0$. Then, by the fluctuation Lemma \cite[Lemma A.1]{Smith} there exists
a divergent sequence $t_n$ such that $u(t_{n})\to M$.

Moreover, the assumption $\mathcal M[\mathfrak a]=0$ implies $\int_s^t\mathfrak a(\xi)\, d\xi \ge -\int_0^T|\mathfrak a(\xi)|\, d\xi$ for all $s\le t$. Thus, if $u(\cdot)$ is a solution of (\ref{perturbada+}) with nonnegative initial condition, then $u(t)>0$ for all $t>0$ and, for arbitrary $\tau_0\ge 0$, 
$$u(t)\ge \int_{\tau_0}^t e^{\int_s^t[\mathfrak a(\xi) +\psi(\xi)]\, d\xi}\mathfrak b(s)\, ds\ge c  \int_{\tau_0}^t e^{\int_s^t \psi(\xi)\, d\xi} \, ds
$$
for some positive constant $c$ depending only on $\mathfrak a$ and $\mathfrak b$. Next, fix $\tau_0$ such that $|\psi(t)|<\varepsilon$ for $t>\tau_0$ with $0<\varepsilon< \frac cM$. 
Taking $t=t_n$, it is seen that
$$
u(t_{n})\ge c \int_{\tau_0}^{t_n} e^{-\varepsilon(t_{n}-s)}\, ds = \frac c\varepsilon \left( 1 - e^{-\varepsilon(t_n-\tau_0)}\right),
$$
which is a contradiction when $n$ is sufficiently large. 
\end{proof}

\begin{remark}
\label{OQ}
The above result cannot be extended to the almost periodic case since, as is well known, the assumption $\mathcal M[\mathfrak a]=0$ does not imply that 
the integral $\int_s^t\mathfrak a(\xi)\, d\xi$ is bounded from below for arbitrary $t> s$.
In consequence, to the best of our knowledge, to determine if always $\liminf\limits_{t\to +\infty}u(t)=+\infty$ remains an elusive question. 
\end{remark}

\subsection{Some weak extinction examples}
The following results explore the asymptotic behavior of the metapopulation for the limit cases $\mathcal{R}_{0}=h$ and $\mathcal{R}_{1,1}^{*}=h$ for $n=1$ and $n=2$.

\begin{theorem}
\label{wextinction-1}
If the fraction of habitable patches together with the colonization and extinction rates are such that $h=\mathcal{R}_{0}$,
then the fraction of occupied patches by the first species verifies a weak extinction, namely, any solution of \eqref{EEN1}
with $p_{1}(0)\in (0,h)$ verifies $\liminf\limits_{t\to +\infty}p_{1}(t)=0$.
\end{theorem}

\begin{proof}
Let us recall that  $u=1/p_{1}$ transforms (\ref{EEN1}) into (\ref{lineaire}): 
    \begin{equation*}
    u'= a_{11}(t)u + c_{1}(t) \quad \textnormal{where $a_{11}(t)= -[c_1(t)h-m_1(t)]$}.
    \end{equation*}

As $p_{1}(t)<h\leq 1$, it follows that $u(t)>1$ for any $t \geq 0$. Now, we will prove that 
$t\mapsto u(t)$ is upper unbounded on $[0,+\infty)$. Indeed, otherwise, 
by using (\ref{infap})  there exists $c_{0}>0$ such that $\liminf\limits_{t\to +\infty}c_{1}(t)\geq c_{0}>0$ and, by the statement 4 from
Proposition \ref{Prop-Cop}, 
the linear equation $x'=a_{11}(t)x$ has an exponential dichotomy
on $[0,+\infty)$ and by Proposition \ref{AVED} we would have that $\mathcal{M}[a_{11}]\neq 0$. Nevertheless, this
contradicts the assumption $h=\mathcal{R}_{0}$ which is equivalent to $\mathcal{M}[a_{11}]=0$. 

Finally, as any solution $u(\cdot)$ of (\ref{lineaire}) is upper unbounded on $[0,+\infty)$, we have that 
$$
\limsup\limits_{t\mapsto +\infty}u(t)=+\infty  \quad \textnormal{or equivalently} \quad \liminf\limits_{t\to +\infty}p_{1}(t)=0,
$$
and the result follows.
\end{proof}

\begin{theorem}[weak extinction of the inferior competitor when $n=2$]
\label{WBCP}
    Assume that $\mathcal{M}[c_{1}]<\mathcal{M}[c_{2}]$ is verified. If the fraction of habitable patches together with the colonization and extinction rates are such that the properties
    \begin{displaymath}
     \mathcal{R}_{0}<h=\mathcal{R}_{1,1}^{*},   
    \end{displaymath}
are also verified, then the first species is strongly uniformly persistent while the second one is weakly
driven to extinction, namely, for every positive solution $t\mapsto (p_1(t),p_2(t))$ of \eqref{tilman1-TV-n2}--\eqref{TH}, 
it follows that 
\begin{displaymath}
\lim\limits_{t\to +\infty}|p_{1}(t)-p_{1,1}^{*}(t)|=0  \quad \textnormal{and} \quad \liminf\limits_{t\to +\infty}p_{2}(t)=0,
\end{displaymath}
where $t\mapsto p_{1,1}^{*}(t)$ is defined by \eqref{sol-periodique}.
\end{theorem}

\begin{proof}

Since $\mathcal{R}_0<h$, the proof that $|p_{1}(t)-p_{1,1}^{*}(t)|=\textit{o}(1)$ and the convergence is exponential can be made in a similar way as in Theorem \ref{Prop-1}.

In addition, and similarly as in the proof of Theorem \ref{BCP}, the second equation of (\ref{tilman1-TV-n2}) can be transformed 
via $v=1/p_{2}$ in the inhomogeneous equation (\ref{a2-bis}):
\begin{equation*}
v'=[a_{21}(t)+\phi(t)]v+c_{2}(t), 
\end{equation*}
where $t\mapsto a_{21}(t)$ is defined by (\ref{a2}) and
$\phi(t)=[c_1(t)+c_{2}(t)]\{p_1(t)-p_{1,1}^{*}(t)\}$, which is convergent to zero since $p_{1}(t)-p_{1,1}^{*}(t)=\textit{o}(1)$
and the colonization rates are bounded.

As $h=\mathcal{R}_{1,1}^{*}$, this implies that $\mathcal{M}[a_{21}]=0$. Moreover, as $\phi(t)=\textit{o}(1)$, Lemma \ref{L33} implies that any solution $v(\cdot)$ of (\ref{a2-bis}) is upper unbounded on $[0,+\infty)$, that is 
$$
\limsup\limits_{t\mapsto +\infty}v(t)=+\infty,  \quad \textnormal{which implies} \quad \liminf\limits_{t\to +\infty}p_{2}(t)=0,
$$
and the result follows.
\end{proof}

\begin{theorem}[Weak extinction of the advantaged competitor when $n=2$]
\label{T63}
Assume that $\mathcal{M}[c_{1}]<\mathcal{M}[c_{2}]$.
 If the fraction of habitable patches together with the colonization and extinction rates are such that the inequalities
    \begin{displaymath}
       \mathcal{R}_{1}<h= \mathcal{R}_{0},
    \end{displaymath}
are verified, then any solution $t\mapsto (p_{1}(t),p_{2}(t))$ of the system \eqref{tilman1-TV-n2}--\eqref{TH} verifies:
    \begin{displaymath}
\liminf\limits_{t\to +\infty}p_{1}(t)=0, 
\end{displaymath}
and the system $x'=a_{22}(t)x$, where $a_{22}(t)$ is defined by \eqref{a3}, has exponential dichotomy  with identity projector and constants  $K\geq 1$ and $\alpha>0$. 

In addition, if $\limsup\limits_{t\to+\infty}[c_1(t)+c_2(t)]p_1(t)=a_3^+<\alpha$ then the second species is strongly uniformly persistent.
\end{theorem}

 \begin{proof}
Since $h=\mathcal{R}_{0}$, this implies that $\mathcal{M}[a_{11}]=0$, and by Theorem \ref{wextinction-1} it follows that $\liminf\limits_{t\to +\infty}p_{1}(t)=0$. Moreover, since $h>\mathcal{R}_1$ is equivalent to $\mathcal{M}[a_{22}]<0$, we conclude by 
Proposition \ref{AVED} that system $x'=a_{22}(t)x$ has an exponential dichotomy on $\mathbb{R}$ with the identity projector and constants $K\geq 1$ and $\alpha>0$.

In order to study the second species, let us consider the inhomogeneous equation
\begin{displaymath}
\dot{u}=b_{22}(t)u+c_{2}(t) \,\, \textnormal{where} \,\, b_{22}(t)=a_{22}(t)+[c_{1}(t)+c_{2}(t)]p_{1}(t),    
\end{displaymath}
corresponding to the second equation of (\ref{tilman1-TV-n2}). If $\limsup\limits_{t\to+\infty}[c_1(t)+c_2(t)]p_1(t)=a_3^+<\alpha$ holds, then the Proposition \ref{rough-a} says that $z'=[a_{22}(t)+\psi(t)]z$ also has an exponential dichotomy on $[0,+\infty)$ with the identity projector, that is, there exist $K\geq 1$ and $\alpha-a_3^+= \alpha_0>0$. As consequence, Lemma \ref{LT1} implies the existence of $\theta>0$ such that
$$
\limsup\limits_{t\to +\infty}u(t)\leq \theta
\quad \textnormal{or equivalently} \quad
\liminf\limits_{t\to +\infty}p_{2}(t)\geq \theta^{-1},
$$
and the uniform persistence of $p_{2}$ follows.
\end{proof}

\begin{remark}
If the rates $c_{i}(\cdot)$ and $m_{i}(\cdot)$ are $T$--periodic functions,
for $i=1,2$ the above results can be improved since the Lemma \ref{Lper} allows to prove that
the weak persistence cannot be verified.
\end{remark}

\subsection{An open question: existence of weakly persistent solutions}
The above results only provide sufficient conditions ensuring weak extinction of a species and this is due
to Lemma \ref{L33}, which ensures the unboundedness of the solutions $t\mapsto u(t)$ of (\ref{perturbada+}). Nevertheless,
it does not ensure that $\lim\limits_{t\to +\infty}u(t)=+\infty$ since, as stated by Remark \ref{OQ}, the case $\liminf\limits_{t\to +\infty}u(t)=\ell>1$
is theoretically possible and, as $p(t)=1/u(t)$ it implies that a species $t\mapsto p(t)$ will verify
$$
\liminf\limits_{t\to +\infty}p(t)=0 \quad \textnormal{and} \quad \limsup\limits_{t\to +\infty}p(t)=1/\ell>0,
$$
that is, a species could be weakly persistent and driven weakly to the extinction at the same time. 

The above comment combined with Remark \ref{OQ} raise a natural question: in absence of exponential dichotomy, can a species be weakly persistent?
This is equivalent to determine additional conditions to those stated in Lemma \ref{L33} ensuring that some solutions of (\ref{perturbada+}) have positive lower limit at $+\infty$. 

\section{Numerical examples}

To illustrate our results, we analyze the dynamical behavior of the system
(\ref{tilman1-TV-n2}) and we will consider the almost periodic colonization rates given by:
\begin{displaymath}
\begin{array}{rcl}
    c_1(t)&=& 2.1+0.3\sin(\pi t )-0.3\cos(\sqrt{2}\pi t),\\
    c_2(t)&=& 3-0.3\sin(\sqrt{2}\pi t)+0.5\cos(\sqrt{3}\pi t),
\end{array}
\end{displaymath}
and we point out that these functions are not supported by experimental data. We can verify that
\begin{displaymath}
\mathcal{M}[c_{1}]=  2.1  \quad \textnormal{and} \quad \mathcal{M}[c_{2}]= 3.
\end{displaymath}

The parameters used in the simulation for (\ref{tilman1-TV-n2}) are summarized in Table \ref{table:1};
\begin{table}[h]
\centering
\begin{tabular}{||l l l l ||} 
 \hline
 Parameter & Values & Meaning & Source \\ 
 \hline
 $m_1$ & $0.7$ & mortality rate of species 1 & this work \\
 $m_2$ & $0.2$  & mortality rate of species 2 & this work \\ 
 $h$ & $[0.03\,,\,0.78]$ & percentage of habitat available & this work\\[1ex] 
 \hline
\end{tabular}
\caption{Synthetic parameters used in the simulations for  (\ref{tilman1-TV-n2}).}
\label{table:1}
\end{table}

Then we have 
\begin{displaymath}
\mathcal{R}_{0}=\frac{\mathcal{M}[m_{1}]}{\mathcal{M}[c_{1}]}=0.\bar{3} \quad   \mbox{and} \quad
\mathcal{R}_{1}=\frac{\mathcal{M}[m_{2}]}{\mathcal{M}[c_{2}]}=0.0\bar{6}.    
\end{displaymath}

The numerical simulations were implemented using Python 3, and the libraries \textsc{numpy}, \textsc{scipy}, and \textsc{matplotlib.pyplot} for efficient computation and visualization\footnote{The corresponding code is accessible at \url{https://github.com/Oehninger/Metapopulations}.}. In particular, the function \textsc{odeint} of \textsc{scipy.integrate} was used to solve the system (\ref{tilman1-TV-n2}) over the interval $[0,1000]$. Next, we approximate $p_{1,1}^{*}(\cdot)$, using the values of $p_1(\cdot)$ in the interval $[100,1000]$ and cubic interpolation. Then we consider $\mathcal{M}[(c_{1}+c_{2})p_{1,1}^*]\approx \frac{1}{900}\int_{100}^{1000}[c_{1}(t)+c_{2}(t)]p_{1,1}^*(t)dt$, and use this value to approximate $\mathcal{R}_{1,1}^*$. 
The above process is repeated for each $h\in \{0.78,\, 0.63,\, 0.48,\, 0.\bar{3},\, 0.18,\, 0.03\}$, in order to represent different behavior of the solutions and illustrate our results.

\medskip

The Figure \ref{fig: a} considers a fraction of habitable patches $h=0.78$. In this case we have that $\mathcal{R}_{0}=0,\bar{3}<h$ and, by Theorem \ref{Prop-1}, it follows that the fraction of patches 
occupied by the species 1 converges exponentially to an almost periodic function $t\mapsto p_{1,1}^{*}(t)$, which leads to
a value $\mathcal{R}_{1,1}^{\ast} \approx 0.82606$.  Finally, as $\mathcal{R}_{0}<h<\mathcal{R}_{1,1}^{\ast}$ Theorem \ref{BCP} implies the 
strong extinction of the species 2.

\medskip

The Figure \ref{fig: b} considers a fraction of habitable patches $h=0.63$. In this case we also have that $\mathcal{R}_{0}=0,\bar{3}<h$ and, by Theorem \ref{Prop-1}, it follows that the fraction of patches 
occupied by the species 1 converges exponentially to an almost periodic function $t\mapsto p_{1,1}^{*}(t)$, which leads to $\mathcal{R}_{1,1}^{\ast} \approx 0.57113$.  Finally, as $\max\{\mathcal{R}_{0},\mathcal{R}_{1,1}^{\ast}\}<h$, the Theorem \ref{GP} implies the strong persistence of the species 2.

\medskip

The Figure \ref{fig: c} considers a fraction of habitable patches $h=0.48$. In this case we also have that $\mathcal{R}_{0}=0,\bar{3}<h$ and, by Theorem \ref{Prop-1}, it follows that the fraction of patches 
occupied by the species 1 converges exponentially to an almost periodic function $t\mapsto p_{1,1}^{*}(t)$, which leads to
a value $\mathcal{R}_{1,1}^{\ast} \approx 0.31611$.  Finally, as $\max\{\mathcal{R}_{0},\mathcal{R}_{1,1}^{\ast}\}<h$, the Theorem \ref{GP} implies the strong persistence of the species 2. Furthermore, we stress that in this scenario the almost periodic fraction $t\mapsto p_{2,1}^{*}(t)$ of the patches occupied by the weaker competitor surpasses those of the advantaged one $t\mapsto p_{1,1}^{*}(t)$.

\medskip

The Figure \ref{fig: d} considers a fraction of habitable patches $h=0.\bar{3}$. In this case we have that $\mathcal{R}_{0}=h$ and, by Theorem \ref{wextinction-1}, it follows that the fraction of patches 
occupied by the species 1 drives weakly to extinction, namely $\liminf\limits_{t\to+\infty}p_{1}(t)=0$, which leads to
a value $\mathcal{R}_{1,1}^{\ast} \approx 0.06871$.  Despite that $h>\mathcal{R}_{1}$, the Theorem \ref{T63} cannot be applied, but the numerical simulations suggest that the fraction of inhabitated patches by the second species is strongly persistent.

The Figure \ref{fig: e} considers a fraction of habitable patches $h=0.18$. In this case we have that $\mathcal{R}_{0}>h$ and, by Theorem \ref{Prop-1}, it follows that the fraction of patches 
occupied by the species 1 drives strongly to extinction $p_{1}(t)=\textit{o}(1)$, which leads to
a value $\mathcal{R}_{1,1}^{\ast} \approx 0.06666$.  Finally, as $\mathcal{R}_{1}=0.0\bar{6}$, we have  $\mathcal{R}_1<h <\mathcal{R}_0$ and  the Theorem \ref{FSP} implies the strong persistence of species 2.

The Figure \ref{fig: f} considers a fraction of habitable patches $h=0.03$. In this case we have that $\mathcal{R}_{0}>h$ and, by Theorem \ref{Prop-1}, it follows that the fraction of patches 
occupied by the species 1 drives strongly to extinction $p_{1}(t)=\textit{o}(1)$, which leads to
a value $\mathcal{R}_{1,1}^{\ast} \approx 0.06666$.  Moreover, as $\mathcal{R}_{1}=0.0\bar{6}>h$,  by Theorem \ref{Prop-1}, it follows that the fraction of patches 
occupied by the species 2 also drives strongly to extinction $p_{2}(t)=\textit{o}(1)$.

\begin{figure}[ht!]
    \centering
\begin{subfigure}{0.355\textwidth}
    \includegraphics[width=\textwidth]{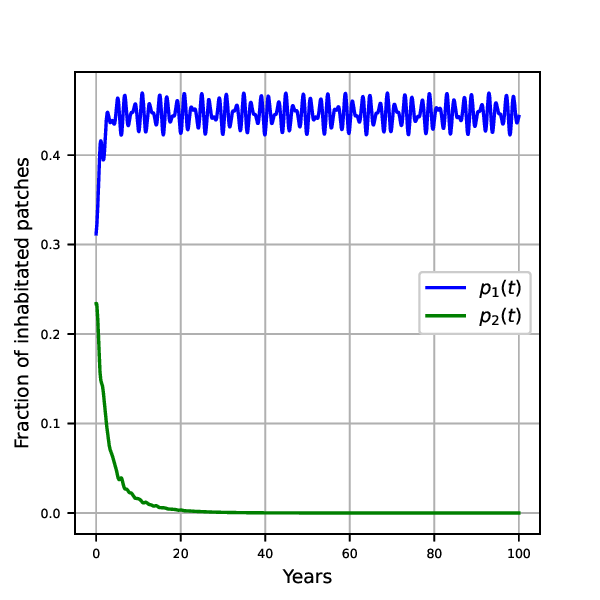} 
    \caption{\scriptsize{Dynamics with $h=0.78$ and $\mathcal{R}_{1,1}^*\approx 0.82606$.}}\label{fig: a}
\end{subfigure}
\centering
\begin{subfigure}{0.355\textwidth}
    \includegraphics[width=\textwidth]{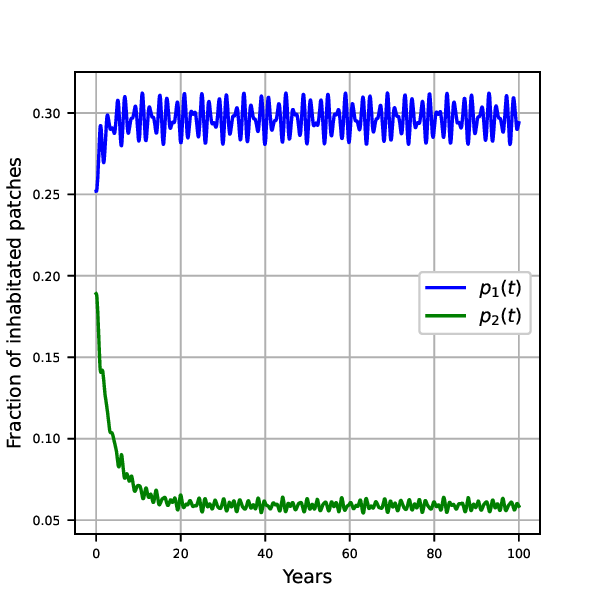} 
    \caption{\scriptsize{Dynamics with $h=0.63$ and $\mathcal{R}_{1,1}^*\approx 0.57113$.}}\label{fig: b}
\end{subfigure}
\centering
\begin{subfigure}{0.355\textwidth}
    \includegraphics[width=\textwidth]{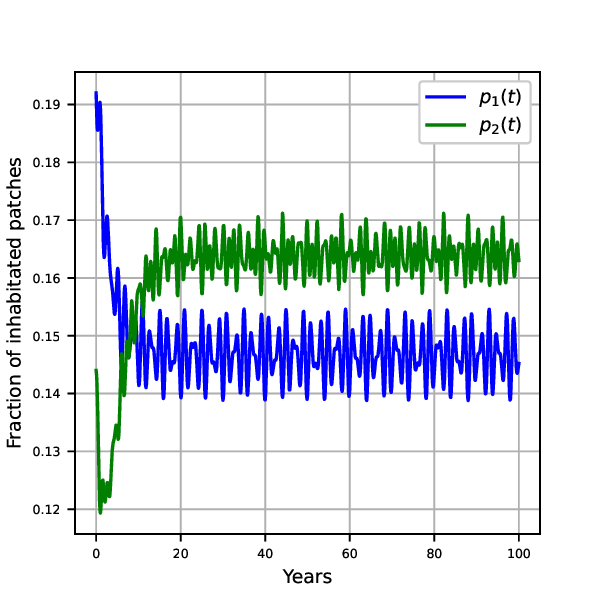} 
    \caption{\scriptsize{Dynamics with $h=0.48$ and $\mathcal{R}_{1,1}^*\approx 0.31611$.}}\label{fig: c}
\end{subfigure}
\centering
\begin{subfigure}{0.355\textwidth}
    \includegraphics[width=\textwidth]{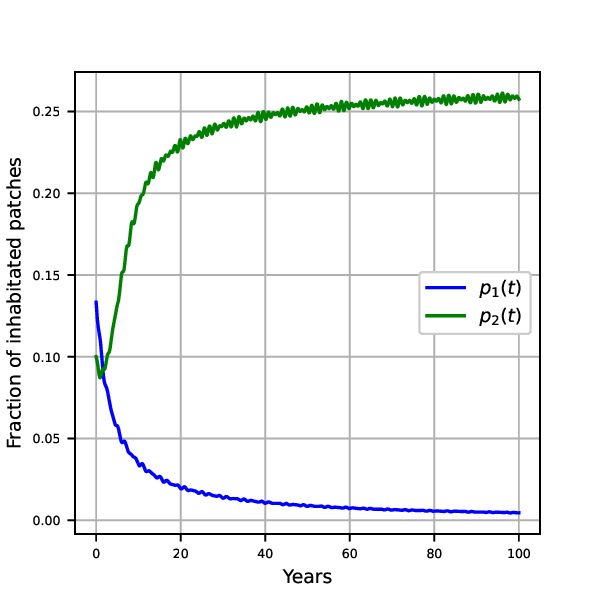} 
    \caption{\scriptsize{Dynamics with $h=0.\overline{3}$ and $\mathcal{R}_{1,1}^*\approx 0.06871$.}}\label{fig: d}
\end{subfigure}
\centering
\begin{subfigure}{0.35\textwidth}
    \includegraphics[width=\textwidth]{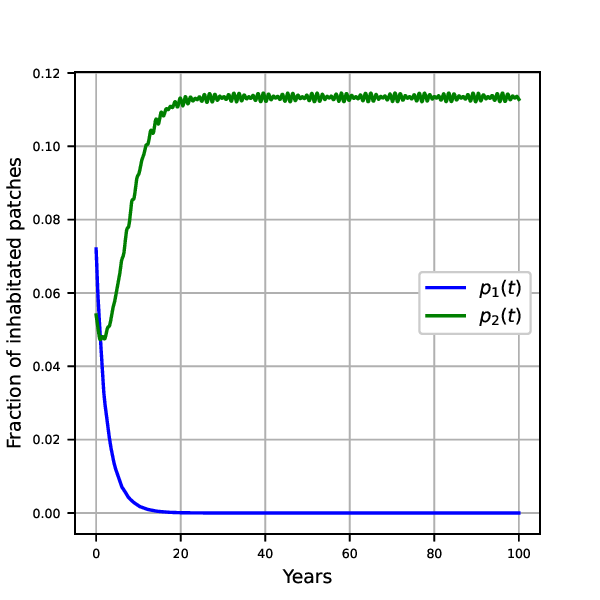} 
    \caption{\scriptsize{Dynamics with $h=0.18$ and $\mathcal{R}_{1,1}^*\approx 0.06666$.}}\label{fig: e}
\end{subfigure}
\centering
\begin{subfigure}{0.35\textwidth}
    \includegraphics[width=\textwidth]{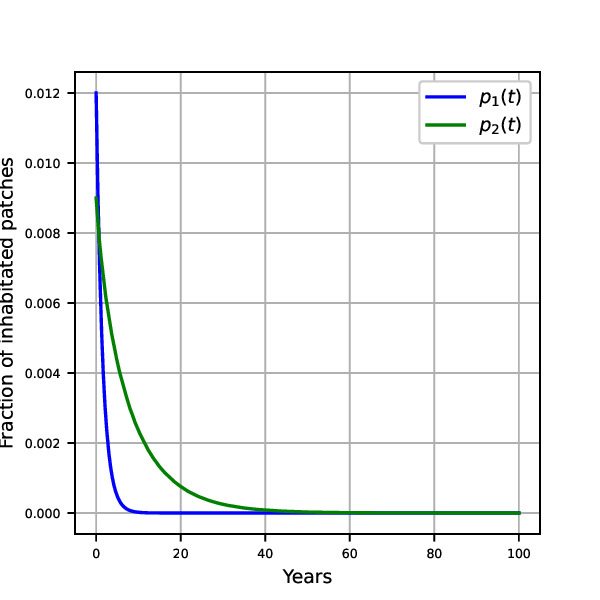} 
    \caption{\scriptsize{Dynamics with $h=0.03$ and $\mathcal{R}_{1,1}^*\approx 0.06666$.}}\label{fig: f}
\end{subfigure}

\caption{\scriptsize{Dynamic of the fraction of inhabited patches by the species for different $h$. The initial conditions are:  $(p_1,p_2)= (0.4h, 0.3h)$, while $\mathcal{R}_0=1/3$ and $\mathcal{R}_1=2/30$.}}\label{fig: 2}
\end{figure}

\newpage

\newpage
\appendix
\section{Proof of Technical Results}

\subsection{Proof of Lemma \ref{L1}}

\begin{proof}
The positiveness of $p_{i}(t)$ for any $t$ is straightforward. Now, let us define $u(t):=p_{1}(t)+\ldots+p_{n}(t)$ and notice that
(\ref{tilman1-TV}) can be written as follows:
\begin{displaymath}
  p_{i}'=c_i(t) p_i\left(h-u(t)+\sum_{j=i+1}^{n}p_j \right)-m_i(t)p_i-\sum_{j=1}^{i-1}c_j(t) p_i p_j, \quad
\textnormal{for $i=1,\ldots,n$}.  
\end{displaymath}

The proof will be made by contradiction. Indeed, if we suppose that $u(t_1)=h$ for the first time at some $t_1>t_0$, then $u'(t_1)\ge 0$ and we may write
\begin{displaymath}
p_i'(t_{1}) =c_i(t_{1}) p_i(t_{1})\sum_{j=i+1}^{n}p_j(t_{1}) -m_i(t_{1})p_i(t_{1})-\sum_{j=1}^{i-1}c_j(t_{1}) p_i(t_{1}) p_j(t_{1})  \quad
\textnormal{for $i=1,\ldots,n$}, 
\end{displaymath}
which combined with the positiveness of $m_{i}(t_{1})p_{i}(t_{1})$ yields to:
$$
\begin{array}{rcl}
u'(t_1) &< & \displaystyle   \sum_{i=1}^n \sum_{j=i+1}^{n} c_i(t_{1}) p_i(t_{1}) p_j(t_{1})  - 
\sum_{i=1}^n\sum_{j=1}^{i-1}c_j(t_{1}) p_i(t_{1}) p_j(t_{1})  \\\\
&=& \displaystyle   \sum_{i=1}^{n-1} \sum_{j=i+1}^{n} c_i(t_{1}) p_i(t_{1}) p_j(t_{1})  - 
\sum_{i=2}^n\sum_{j=1}^{i-1}c_j(t_{1}) p_i(t_{1}) p_j(t_{1}), 
\end{array}
$$
where the last equality holds because   $\sum_{n+1}^{n}=\sum_{j=1}^0=0$. Next, we may change the order of summation in the first term in order to obtain
$$u'(t_1)<\sum_{j=2}^{n} \sum_{i=1}^{j-1} c_i(t_{1}) p_i(t_{1}) p_j(t_{1})  - 
\sum_{i=2}^n\sum_{j=1}^{i-1}c_j(t_{1}) p_i(t_{1}) p_j(t_{1}) 
$$
and by relabeling the sub-indices we deduce that $u'(t_1)<0,$
a contradiction.
\end{proof}

\subsection{Proof of Lemma \ref{AHIH}}

We will provide the proof for the interval $[0,+\infty)$ since the other case can be done in an analogous way.
As the backward direction is trivial, we only prove the forward direction by assuming that (\ref{dico}) 
is verified on $[T,+\infty)$.

If $0\leq s<T\leq t$, we have that:
\begin{displaymath}
\begin{array}{rcl}
\left|\int_s^t \mathfrak{a}(r)\,dr 
\right|  &\geq & 
\left|\int_T^t \mathfrak{a}(r)\,dr 
\right|-
\left|\int_s^T \mathfrak{a}(r)\,dr 
\right|  \\\\
&\geq & 
\alpha(t-T)-C-\max\limits_{\theta \in [0,T]}
\left|\int_\theta^T \mathfrak{a}(r)\,dr 
\right|  \\\\
&\geq & \alpha(t-s)-C-\max\limits_{\theta \in [0,T]}
\left|\int_\theta^T \mathfrak{a}(r)\,dr\right| +\alpha(s-T) \\\\
&\geq & \alpha(t-s)-\underbrace{\left\{C+\max\limits_{\theta \in [0,T]}
\left|\int_\theta^T \mathfrak{a}(r)\,dr\right| +\alpha T \right\}}_{:=C_{1}}.
\end{array}
\end{displaymath}

If $0\leq s\leq t<T$, the above case implies that $\left|\int_s^T \mathfrak{a}(r)\,dr 
\right|\geq \alpha(T-s)-C_{1}$ and we can deduce that:
\begin{displaymath}
\begin{array}{rcl}
\left|\int_s^t \mathfrak{a}(r)\,dr 
\right|  &\geq & 
\left|\int_s^T \mathfrak{a}(r)\,dr 
\right|-
\left|\int_t^T \mathfrak{a}(r)\,dr 
\right|  \\\\
&\geq & \alpha (T-s)-C_{1}-\max\limits_{\theta \in [0,T]}
\left|\int_\theta^T \mathfrak{a}(r)\,dr\right| \\\\
&\geq & \alpha (t-s)-\underbrace{\left\{C_{1}+\max\limits_{\theta \in [0,T]}
\left|\int_\theta^T \mathfrak{a}(r)\,dr\right|\right\}}_{:=C_{2}}.
\end{array}
\end{displaymath}

By gathering the above estimations, and noticing that $C\leq C_{1}\leq C_{2}$, we can conclude that:
\begin{displaymath}
\left|\int_s^t \mathfrak{a}(r)\,dr 
\right|\ge \alpha(t-s) - C_{2}\qquad \textnormal{for any $t\geq s$ with $t,s\in [0,+\infty)$}.
\end{displaymath}
 
\subsection{Proof of Proposition \ref{rough-a}}

    Fix $\alpha_0>0$ such that $\alpha_{0}< \alpha - \delta_{0}$ and $T\geq 0$ such that $|\mathfrak{e}(t)|<\alpha-\alpha_0$ for any $t\ge T$. The proof will consider
the three possible cases: $t\geq s\geq T$, $T\geq t\geq s$ and $t\geq T \geq s$.
    
\medskip 

Firstly, when $t\ge s\ge T$ it is verified that: 
    $$\left|\int_s^t [\mathfrak{a}(r)+\mathfrak{e}(r)] \,dr \right| 
\ge \alpha(t-s) - C - \int_s^t |\mathfrak{e}(r)|\, dr \ge \alpha_0(t-s) - C.
$$

Secondly, when $s\leq t\leq T$, let us define $C_0:= \int_0^{T} |\mathfrak{a}(r)+\mathfrak{e}(r)|\ dr$. It is clear that 
$$\left|\int_s^t [\mathfrak{a}(r)+\mathfrak{e}(r)] \, dr \right| 
\ge - C_0  \ge \alpha_0(t-s) - C_0 -\alpha_0\,T.
$$
Finally, for $t> T\ge s$, we have that
\begin{displaymath}
\begin{array}{rcl}
\left|\int_s^t [\mathfrak{a}(r)+\mathfrak{e}(r)] \,dr \right| &\ge&  
\left|\int_{T}^t [\mathfrak{a}(r)+\mathfrak{e}(r)] \,dr \right| - 
\int_s^{T} |\mathfrak{a}(r)+\mathfrak{e}(r)|\ dr \\
&\ge& \alpha_0(t-T) - C-C_0\ge \alpha_0 (t-s) - \alpha_0 T-
C-C_0,
\end{array}
\end{displaymath}
then we have that the equation has an exponential dichotomy with constants $\alpha_0$ and $\tilde C:= C+ C_0 + \alpha_0\,T$.

\subsection{Proof of Proposition \ref{Prop-Cop}} Proof of $1) \Rightarrow 2)$:
If (\ref{proj-id}) holds, then any solution of (\ref{sca0}) has the form:
\begin{equation}
\label{sol-gral}x(t)= x(t_0) e^{\int_{t_0}^{t}\mathfrak{a}(r)\,dr} + \int_{t_0}^t \mathfrak{v}(s) e^{\int_{s}^{t}\mathfrak{a}(r)\,dr}\, ds.    
\end{equation}
The first term of the right-hand side is bounded since it converges to $0$ when $t\to +\infty$, while the second term satisfies, for $t\ge t_0$:
$$\left|\int_{t_0}^t \mathfrak{v}(s) e^{\int_{s}^{t}\mathfrak{a}(r)\,dr}\, ds\right| \le \|\mathfrak{v}\|_\infty 
\int_{t_0}^t Ke^{-\alpha(t-s)}\, ds = \frac K\alpha \|\mathfrak{v}\|_\infty \left(1 - e^{-\alpha (t-t_0)}\right),
$$
and the boundedness on $[0,+\infty)$ of any solution of (\ref{sca0}) is verified. Now, if (\ref{proj-null}) holds instead, then it is readily verified that the expression
$$
x(t)= -\int_t^{+\infty} \mathfrak{v}(s) e^{-\int_{t}^{s}\mathfrak{a}(r)\,dr}\, ds
$$
is well defined and gives an explicit description for  the unique solution of (\ref{sca0})
which is bounded on $[0,+\infty)$. 

On the other hand, observe that the statements 3) and 4) are particular cases of 2), which implies that $1) \Rightarrow 3)$ and $1) \Rightarrow 4)$ are straightforwardly verified.

Proof of $4) \Rightarrow 3)$: Let $t\mapsto x(t)$ be a bounded solution of (\ref{inhp}) such that
\begin{displaymath}
m\leq x(t)\leq M \quad \textnormal{for any $t\in [0,+\infty)$},    
\end{displaymath}

By hypothesis, we have the existence of $v_{0}>0$ such that $|\mathfrak{v}(t)|>v_{0}$
for any $t\geq t_{0}$ with $t_{0}\geq 0$. Without loss of generality, we can suppose
$t_{0}=0$.

Firstly, we will assume that $\mathfrak{v}(t)\ge v_0 >0$ for any $t\geq 0$, 
then
the above solution $t\mapsto x(t)$ also verifies:
$$
x(t)=x_0 e^{\int_{0}^{t}\mathfrak{a}(r)\,dr } + \int_{0}^t e^{\int_{s}^{t}\mathfrak{a}(r)\,dr}\mathfrak{v}(s)ds\ge x_0 e^{\int_{0}^{t}\mathfrak{a}(r)\,dr} + v_0\int_{0}^{t} e^{\int_{s}^{t}\mathfrak{a}(r)\,dr}\,ds,
$$
which allow us to deduce:
$$
\frac {x(t)}{v_0} \ge 
\frac{x_0}{v_0} e^{\int_{0}^{t}\mathfrak{a}(r)\,dr} + \int_{0}^{t} e^{\int_{s}^{t}\mathfrak{a}(r)\,dr}\,ds:=z(t).
$$

We can see that $t\mapsto z(t)$ is solution of $z'=\mathfrak{a}(t)z+1$ passing through $x_{0}/v_{0}$ at $t=0$, which also verifies $z(t)\leq M/v_{0}$ for any $t\geq 0$.

Now, let us observe that
\begin{displaymath}
\begin{array}{rcl}
z(t)&=&\displaystyle \frac{x_0}{v_0} e^{\int_{0}^{t}\mathfrak{a}(r)\,dr} + \int_{0}^{t} e^{\int_{s}^{t}\mathfrak{a}(r)\,dr}\,ds \\\\
&=&\displaystyle ||\mathfrak{v}||_{\infty}^{-1}\left(x_{0}\frac{||\mathfrak{v}||_{\infty}}{v_0} e^{\int_{0}^{t}\mathfrak{a}(r)\,dr} + \int_{0}^{t} e^{\int_{s}^{t} \mathfrak{a}(r)\,dr}||\mathfrak{v}||_{\infty}\,ds \right)\\\\
&\geq & \displaystyle ||\mathfrak{v}||_{\infty}^{-1}\left(x_{0}\ e^{\int_{0}^{t}\mathfrak{a}(r)\,dr} + \int_{0}^{t} e^{\int_{s}^{t}\mathfrak{a}(r)\,dr}\mathfrak{v}(s)\,ds \right)\\\\
&\geq & ||\mathfrak{v}||_{\infty}^{-1}x(t),
\end{array}
\end{displaymath}
and we proved that 
\begin{displaymath}
m||\mathfrak{v}||_{\infty}^{-1}\leq z(t)\leq v_{0}^{-1}M \quad \textnormal{for any $t\in [0,+\infty)$}.   \end{displaymath}

{Secondly, if we assume that $-\mathfrak{v}(t) \geq v_{0}>0$ for any $t\geq 0$, we multiply (\ref{inhp}) by $-1$ and make the transformations $u=-x$ and $\mathfrak{w}(t)=-\mathfrak{v}(t)$, leading to $u'=\mathfrak{a}(t)u+\mathfrak{w}(t)$, which also has a bounded solution
and the above proof can be carried out}.

Proof of $3) \Rightarrow 1)$: Let $t\mapsto z(t)$ be a solution
of $z'=\mathfrak{a}(t)z+1$ which is bounded on $[0,+\infty)$. Firstly, a noteworthy
property of the above equation is that: if $z(0)\geq 0$, the solution remains positive
for any $t>0$ while, if $z(0)<0$ the solution either remains negative or has a unique change of sign.

Secondly, let us verify that
\begin{equation}
\label{LIA}
\displaystyle\liminf_{t\to+\infty} |z(t)| >0.
\end{equation}

In fact, this property is immediately verified when $z(\cdot)$ is positive
after some value since $t\mapsto z(t)$ is increasing at any value of $t$ such that $|z(t)|\|\mathfrak{a}\|_\infty < 1$. Now, if $z(t)<0$ for any $t\geq 0$, we have to consider two cases: 

\medskip

\noindent \textit{Case a)} $\limsup\limits_{t\to+\infty} z(t)<0$, which is equivalent to (\ref{LIA}).

\noindent \textit{Case b)} $\limsup\limits_{t\to+\infty} z(t)=0$, which cannot be possible, because $z(\cdot)$ is increasing after some finite time and, consequently, $z(t)=\textit{o}(1)$, which leads to $\liminf\limits_{t\to+\infty} z'(t) >0$, obtaining a contradiction with the boundedness of $t\mapsto z(t)$ and (\ref{LIA}) follows.

In consequence, we can assume that $0<c\le |z(t)|\le M$ for any $t\geq T$.  In the case that $z(\cdot)$ is positive for $t>s\geq T$ it follows that:
$$C\ge \ln\left(\frac{z(t)}{z(s)}\right) = \int_s^t \mathfrak{a}(r)\, dr + \int_s^t \frac 1{z(r)}\, dr,$$
which leads to:
$$\int_s^t \mathfrak{a}(r)\, dr \le  C - \frac 1M (t-s).$$

Similarly, if $z(\cdot)$ is negative  for any $t\leq s$ we will have that
$$
\int_t^s \mathfrak{a}(r)\, dr \ge  -C + \frac 1M (s-t),
$$
then the equation (\ref{sca1}) has an exponential dichotomy on $[T,+\infty)$
and the exponential dichotomy is verified on $[0,+\infty)$ as a consequence of Lemma  \ref{AHIH}.

\subsection{Proof of Proposition \ref{SPEC}}

{We will prove an equivalent result, namely:
$$
\mathbb{R}\setminus \Sigma_{[0,+\infty)]}(\mathfrak{a})=(-\infty,\beta^{-}(\mathfrak{a}))\cup (\beta^{+}(\mathfrak{a}),+\infty).
$$}
{If $\lambda \notin \Sigma_{[0,+\infty)}(\mathfrak{a})$
it follows from Definition \ref{DEFSPEC} that $x'=[\mathfrak{a}(t)-\lambda]x$
has an exponential dichotomy on $[0,+\infty)$, that is either (\ref{proj-id}) or (\ref{proj-null}) holds for $\mathfrak{a}(\cdot)-\lambda$}. 

{The proof will be a consequence of the following statements:
\begin{itemize}
\item[i)] The property  (\ref{proj-null}) holds for $\mathfrak{a}(\cdot)-\lambda$
if and only if $\lambda \in (-\infty,\beta^{-}(\mathfrak{a}))$,
\item[ii)] The property  (\ref{proj-id}) holds for $\mathfrak{a}(\cdot)-\lambda$
if and only if $\lambda \in (\beta^{+}(\mathfrak{a}),+\infty)$.
\end{itemize}
}

Assume for example that (\ref{proj-null}) holds for $\mathfrak{a}(\cdot)-\lambda$, then for arbitrary $s>t$ we have
$$ \frac 1{s-t}\int_t^s \mathfrak{a}(r)\, dr \ge \lambda + \alpha - \frac C{s-t} 
$$
for some $\alpha>0$. Taking lower 
limit for $t\to+\infty$ and $L:=s-t\to +\infty$, we deduce that $\beta^-(\mathfrak{a})\ge \lambda + \alpha>\lambda$. Conversely, for $\beta^-(\mathfrak{a})> \lambda$ we may fix $\alpha\in (0,\beta^-(\mathfrak{a})- \lambda)$
and $s_0, L_0>0$ such that 
$$ 
\frac 1{s-t}\int_t^s \mathfrak{a}(r)\, dr > \lambda + \alpha \quad \textnormal{with} \quad t>t_0 \quad  \textnormal{and} \quad s-t>L_0,  
$$
that is
\begin{equation}
\label{condition}
\int_t^s [\mathfrak{a}(r)-\lambda]\, dr >  \alpha(s-t) \quad \textnormal{with} \quad s-t>L_0 \quad \textnormal{and} \quad t>t_0
\end{equation}

We have to study the cases where the conditions $s-t>L_0$ and $t>t_0$ are not satisfied, namely: a) $0\le s-t\le L_0$, b) $t_{0}\geq s-L_{0}>t$ and c) $s-L_{0}>t_{0}\geq t$. 

Firstly, the cases a) and b) imply that $0\le s-t\le L_0$ and $t\le s\le L_0+t_0$. Then if a) or b) are
satisfied we have that:
$$
\int_t^s [\mathfrak{a}(r)-\lambda]\, dr \ge -\|\mathfrak{a} - \lambda\|_\infty (L_0+s_0)
\ge \alpha (s-t) - (\|\mathfrak{a} - \lambda\|_\infty + \alpha)(L_0+t_0).
$$

Finally, if $t\le t_{0} < s-L_{0}$, then by using (\ref{condition}) we have
\begin{displaymath}
\begin{array}{rcl}
\int_t^s [\mathfrak{a}(r)-\lambda]\, dr &=& 
\int_{t_{0}}^{s} [\mathfrak{a}(r)-\lambda]\, dr + \int_{t}^{t_0} [\mathfrak{a}(r)-\lambda]\, dr \\
&>& \displaystyle \alpha (s-t_0) - \|\mathfrak{a} - \lambda\|_\infty t_0 \\
&\ge&  \displaystyle
\alpha (s-t) - (\|\mathfrak{a} - \lambda\|_\infty +\alpha) t_0.
\end{array}
\end{displaymath}

By gathering the above cases we have that (\ref{proj-null}) is verified with
$\alpha>0$ and $C=(\|\mathfrak{a} - \lambda\|_\infty +\alpha)t_{0}$, then the statement i) is verified.
The statement ii) can be proved in an analogous way.

\end{document}